\documentclass[a4paper]{article}
\usepackage[utf8x]{inputenc}
\usepackage{amsmath, amsfonts,amssymb,mathtools,amsthm,mathrsfs,bbm}
\usepackage{graphicx}
\usepackage{xcolor}
\usepackage{natbib}  
\usepackage{bbding}
\usepackage{csquotes}
\usepackage{tikz}
\usepackage{subfigure}
\definecolor{darkgrey}{rgb}{0.75,0.75,0.75}

\usepackage[rightcaption]{sidecap}
\newtheorem{theorem}{Theorem}
\newtheorem{proposition}{Proposition}[section]
\newtheorem{lemma}[proposition]{Lemma}
\newtheorem{corollary}[proposition]{Corollary}
\newtheorem{definition}[proposition]{Definition}
\theoremstyle{definition}
\newtheorem{remark}[proposition]{Remark}

\makeatletter
\renewcommand{\@fnsymbol}[1]{\ensuremath{%
   \ifcase#1\or 1\or 2\or 3\or
   \mathsection\or \mathparagraph\or \|\or 1\or
   2\or 3 \else\@ctrerr\fi}}
\makeatother

\numberwithin{equation}{section}

\begin{document}
\title{\Large Stochastic equations for a self-regulating gene}
\author{Martin Jansen\thanks{Universit\"atsklinik Freiburg,
    Albert-Ludwigs University of Freiburg, Hugstetter Str. ~55, D--79106 Freiburg, Germany,
    e-mail: martin.jansen@uniklinik-freiburg.de}\,
}
\thispagestyle{empty}

\date{\today} 

\maketitle

\begin{abstract}
  \noindent
Expression of cellular genes is regulated by binding of transcription factors to their promoter, either activating or inhibiting transcription of a gene.
Particularly interesting is the case when the expressed protein regulates its own transcription.
In this paper the features of this self-regulating process are investigated. In the here presented model the gene can be in two states. Either a protein is bound to its promoter or not. The steady state distributions of protein during and at the end of both states are analyzed. Moreover a powerful numerical method based on the corresponding master equation to compute the protein distribution in the steady state is presented and compared to an already existing method. Additionally the special case of self-regulation, in which protein can only be produced, if one of these proteins is bound to the promoter region, is analyzed. Furthermore a self-regulating gene is compared to a similar gene, which also has two states and produces the same amount of proteins but is not regulated by its protein-product.
\end{abstract}

\noindent {\bf Keywords and phrases}: Gene expression, chemical reaction network, activated gene, gene regulatory networks, master equation

~

\noindent {\bf AMS Subject Classification}: 92C42 \and 60J28 \and 65C40

\section{Introduction}
Gene expression denotes the cascade of certain reactions in a cell, causing the
synthesis of RNA and proteins originated from a gene. Important steps in this cascade of
reactions are transcription i.e. production of RNA from the corresponding gene by RNA
polymerase and translation i.e. production of protein from RNA
through ribosomes.
The usage of mathematical models for analysis of gene expression is a broad and well
explored field \citep{ay2011mathematical}.
The regulation of gene expression was first described by \cite{jacob1961genetic},
who studied the regulation of lactose metabolism by the lac operon in \textit{Escherichia
Coli}.
A proper regulation of gene expression is an essential feature for a cell. For example
unicellular organisms need to react rapidly on a change of their environment and
multicellular organisms need to control the differentiation of cells.
In a common form of gene regulation proteins (so called regulatory proteins) cause a
change of the binding-rate of DNA polymerase to the promoter region, leading to either
activation or inhibition of transcription of the gene \citep{Ptashne1992}. In this paper
regulatory proteins, directly associated with the product of the corresponding gene
expression are considered. So a self-regulation through a feedback takes places.\newline
There are many different approaches describing genetic regulatory systems such as a
self-regulating gene mathematically. These approaches use for example techniques like 
Boolean networks, ordinary and partial differential equations and directed graphs 
\citep{de2002modeling}. In this context a gene regulated by its own protein-product is
a relative simple system. \newline
Stochasticity plays an important role in gene regulation
\citep{mcadams1999sa,elowitz2002stochastic,paulsson2005models}. Especially
the promoter fluctuation, therefore the stochastic activation and deactivation of a gene
contributes to the cellular noise \citep{zhu2008delay}. This is caused amongst others by
the low number of gene copies, which equals in diploid organisms in general two, and the fact that the gene
is situated at the beginning of the reaction cascade. So proteins easily inherit
stochastic effects. Additionally there are stochastic effects caused by the low number of
RNA copies.\newline
A common approach to model gene expression is to make the same assumptions as in a chemical 
reaction network \citep{gunawardena2003chemical}.
So each reaction underlies mass-action kinetics, which simplifies the analysis a lot.
Hence it is assumed that the protein production steps are exponentially distributed, with
rates depending on the state of the gene.\newline
In the here presented model the self-regulating gene has
two states: either the promoter region of the gene is bound or unbound by a regulatory
protein. Protein number and the state of the gene are the only quantities considered, so
transcription and translation are lumped together as one reaction. 
The model assumptions are the same as used by \cite{hornos2005self}. 
\cite{peccoud1995markovian} considered a similar and simpler model, in which switching of the gene was independent of its product, so no feedback takes places.\newline
We are interested in analysing and computing the steady state distribution of protein numbers during and at the end of the bound and the unbound state.\newline
\cite{hornos2005self} introduced a method to compute the protein distributions during both states. 
As key result another method to compute this distribution is presented and compared with the approach of Hornos.\newline\newline
Outline:
After describing the model, the distribution of the system in the steady state is considered.
Linear dependencies and recursions for this distribution (Theorem \ref{T1}), and a
recursive description of its moments in the unbound and the bound state are derived
(Theorem \ref{T2}). These recursions can be used to compute the correct protein distribution
in the steady state (Theorem \ref{T3}).
Next the distributions of the proteins at the beginning and the end of the bound state
are considered (Theorems \ref{T4}-\ref{T6}).
Hereafter a kind of extreme positive feedback is considered, where no protein is produced
during the unbound state, which always leads to an extinction of the proteins after a
finite number of unbound states. The number and duration of the
bound-unbound cycles until the protein is extinct, given $n$ proteins at time $0$ and
other features, is investigated (Theorem \ref{T7}).
In Section \ref{section_algo} an algorithm based on the recursions of Theorem \ref{T1}, which computes the protein distribution
effectively, is presented and discussed.
In Section \ref{feedback_impact} the influence of feedback on gene expression is analysed by comparing a model without feedback with a model with feedback.
A discussion is stated in Section \ref{section_disc}, in which this method is compared to the one used by \cite{hornos2005self}. Furthermore the impact of stochasticity to the model is illustrated using an example. Proofs are given in Section \ref{section_proof}.
\section{Main results}
In this section a model for a self regulating gene is described and analysed. An algorithm to compute the steady state distribution of the system is deduced.
\subsection{Model and results}\label{section_Main_results}
The model refers to the expression of a single gene. The set of chemical reactions
\begin{align*}
  1: && U & \xrightarrow{\lambda}U+P,\\
  2: && B & \xrightarrow{\tilde\lambda}B+P,\\
  3: && U + P & \xrightarrow{\psi} B,\\
  4: && B & \xrightarrow{\theta} U + P,\\
  5: && P & \xrightarrow{\mu} \emptyset
\end{align*}
is considered for reaction rates $\lambda,  \psi, \theta, \mu>0$ and $\tilde\lambda\geq 0$.
Here, $P$ is a protein which is expressed by a gene, which comes in
two states, $U$ and $B$. Precisely, $U$ and $B$ denote that the gene is unbound or bound
by one protein molecule $P$. If the gene is in state $U$, the protein is
expressed at rate $\lambda$, while it is expressed at rate
$\tilde\lambda$, if the gene is in state $B$. Clearly, the case $\lambda >\tilde\lambda$ ($\lambda < \tilde\lambda$) is a negative (positive)
feedback loop, where a protein inhibits (enhances) its own transcription.\newline
Here the gene produces directly the protein, so the dynamics of mRNA and other intermediate steps in gene expression are not modeled.\newline
Let $N_t$ be the total number of proteins (either free or bound to the
gene) at time $t$, and $A_t=0$ or $A_t=1$ if the promoter is bound to the protein or not. By the law of mass action, the process $(X_t)_{t\geq 0}$ with $X_t = (N_t, A_t)$ obeys the time-change equations
\begin{align*}
  N_t\!  = N_0 +& Y_1\! \left( \int_0^t \lambda 1_{A_s=1} ds\right)\! +\! Y_2\! \left( \int_0^t \tilde\lambda 1_{A_s=0} ds\right)\!
	          -\! Y_3\! \left( \int_0^t  \mu N_s\left( 1-  1_{A_s=0}\right) ds\right)\!,\\
  A_t\! = A_0 -& Y_4\!\left( \int_0^t \psi 1_{A_s=1} N_s ds\right)\! +\! Y_5\! \left( \int_0^t \theta 1_{A_s=0} ds\right)
\end{align*}
for independent, unit rate Poisson processes $Y_1,...,Y_5$. The fact that this system of equations has a unique steady state solution $\lim_{t\rightarrow\infty}X_t=X$ with $(N,A):=X$ follows from Markov chain theory \citep{anderson2011continuous}. %Moreover, since all reaction rates are positive and the production rate is bounded, the process $(X_t)_{t\geq 0}$ is an ergodic Markov chain.\newline
\newline
It holds for all $t\geq0$, that $P\left\{X_t=(0,0)\right\}=0$, since in the bound state one protein is always bound to the gene.
The master equation for $n\in\mathbb{N}$ is
\begin{equation}\begin{split}
\label{Master}
 \frac{\partial P\left\{X_t=(n,1)\right\}}{\partial t}=& P\left\{X_t=(n-1,1)\right\}\lambda-P\left\{X_t=(n,1)\right\}(n(\psi+\mu)+\lambda)\\&+P\left\{X_t=(n+1,1)\right\}(n+1)\mu+P\left\{X_t=(n,0)\right\}\theta,\\
 \frac{\partial P\left\{X_t=(n,0)\right\}}{\partial t}=& P\left\{X_t=(n-1,0)\right\}\tilde\lambda-P\left\{X_t=(n,0)\right\}(\theta+(n-1)\mu+\tilde\lambda)\\&+P\left\{X_t=(n+1,0)\right\}n\mu+P\left\{X_t=(n,1)\right\}n\psi.\\
\end{split}\end{equation}
It is straightforward to derive linear relationships between the states in the steady state.

\begin{theorem}[Recursions and equations for the steady state protein distribution]
  \label{T1} 
	Consider the steady state distribution $X$ and let $n\in\mathbb{N}$. The decay of protein, when there are $n$ proteins, equals the production of protein, when there are $n-1$ proteins 
\begin{equation}
\begin{split}\label{E1}&P\left\{X=(n,1)\right\}\mu n+P\left\{X=(n,0)\right\}\mu(n-1)\\&=P\left\{X=(n-1,1)\right\}\lambda+P\left\{X=(n-1,0)\right\}\tilde{\lambda}.
\end{split}
\end{equation}
The gene switches as often in the state $A=1$, as it switches in the state $A=0$
\begin{equation}
\begin{split}\label{E2}\theta\sum_{i=0}^{\infty}P\left\{X=(i,0)\right\}=\psi\sum_{i=0}^{\infty}iP\left\{X=(i,1)\right\}.
\end{split}
\end{equation}
A state $(n,1)$ is left as often, as it is entered
\begin{equation}
\begin{split}
\label{E3}&P\left\{X=(n,1)\right\}(n(\psi+\mu)+\lambda)\\& =P\left\{X=(n-1,1)\right\}\lambda+P\left\{X=(n+1,1)\right\}(n+1)\mu+P\left\{X=(n,0)\right\}\theta,
\end{split}
\end{equation}
and a state $(n,0)$ is left as often, as it is entered
\begin{equation}
\begin{split}\label{E4}&P\left\{X=(n,0)\right\}((n-1)\mu+\theta+\tilde{\lambda})\\& =P\left\{X=(n-1,0)\right\}\tilde{\lambda}+P\left\{X=(n+1,0)\right\}\mu n+P\left\{X=(n,1)\right\}\psi n.
\end{split}
\end{equation}
\end{theorem}
With these recursions the first moment of the proteins can be characterised and recursions for the bound and the unbound state can be given:
\begin{corollary}[First moments and recursions for the states $A=1$ and $A=0$]
  \label{C1} The expectation values of $N$, $1_{A=1}N$ and $1_{A=0}N$ are
\begin{equation}\label{E5}\mathbb{E}\left[N\right]=P\left\{A=0\right\}\left(1+\frac{\tilde{\lambda}}{\mu}\right)+P\left\{A=1\right\}\frac{\lambda}{\mu},\end{equation}
\begin{equation}\label{E6}\mathbb{E}\left[1_{A=1}N\right]=P\left\{A=0\right\}\frac{\theta}{\psi},\end{equation}
\begin{equation}\label{E7}
\mathbb{E}\left[1_{A=0}N\right]=P\left\{A=0\right\}\left(1+\frac{\tilde{\lambda}}{\mu}-\frac{\theta}{\psi}\right)+P\left\{A=1\right\}\frac{\lambda}{\mu}.
\end{equation}
Furthermore it holds in the unbound state 
\begin{equation}\label{Startpaar}P\left\{X=(0,1)\right\}=\frac{\mu}{\lambda}P\left\{X=(1,1)\right\}.\end{equation}
It holds that \begin{equation}\label{E02}\frac{\min\left(\lambda,\tilde\lambda+\mu\right)}{\mu}\leq\mathbb{E}\left[N\right]\leq\frac{\max\left(\lambda,\tilde\lambda+\mu\right)}{\mu}.\end{equation}
For $A=1$ and $n\geq 1$ the recursion 
\begin{equation}
\begin{split}\label{E8}
&P\left\{X=(n+2,1)\right\}=\\
&\frac{1}{n(n+2)\mu^2}\left[P\left\{X=(n+1,1)\right\}\mu\left(n((n+1)(\psi+\mu)+\lambda)+(\tilde{\lambda}+\theta)(n+1)\right)\right.\\
&\left.-P\left\{X=(n,1)\right\}\left(\tilde{\lambda}\left(n(\psi+\mu)+\lambda\right)+n\lambda\mu+\lambda\theta\right)+
P\left\{X=(n-1,1)\right\}\lambda\tilde{\lambda}\right]
\end{split}
\end{equation}
holds. For $A=0$ and $n\geq 1$ the recursion 
\begin{equation}
\begin{split}
\label{E9}
&P\left\{X=(n+2,0)\right\}=\\
&\frac{1}{n(n+1)\mu^2}\left[-P\left\{X=(n,0)\right\}\left(n\tilde{\lambda}(\mu+\psi)+\lambda\left((n-1)\mu+\theta+\tilde\lambda\right)\right)\right.\\
&\left.+nP\left\{X=(n+1,0)\right\}\mu\left(n\left(\psi+\mu\right)+\theta+\lambda+\tilde\lambda\right)+P\left\{X=(n-1,0)\right\}\lambda\tilde{\lambda}\right]
\end{split}
\end{equation}
holds.
\end{corollary}
It is possible to compute higher moments of $N$ conditioned on both states in a recursive way.

\begin{theorem}[Recursions for higher moments]
\label{T2} 
Let $s\in\mathbb{N}$, the higher moments of $N$ and $1_{A=1}N$ are given by
\begin{equation}
\begin{split}\label{E10}
&\mathbb{E}\left[1_{A=1}N^{s+1}\right]=\\
&\frac{1}{\psi}\Big(\theta \mathbb{E}\left[1_{A=0}N^s\right]+\sum_{j=1}^{s}\Big(\binom{s}{j}\left((-1)^{j}\mu \mathbb{E}\left[1_{A=1}N^{s-j+1}\right]+\lambda\mathbb{E}\left[1_{A=1}N^{s-j}\right]\right)\Big)\Big),
\end{split}
\end{equation}
\begin{equation}
\begin{split}\label{E11}
&\mathbb{E}\left[N^{s+1}\right]=\\
&\frac{\lambda}{\mu}\mathbb{E}\left[N^{s}1_{A=1}\right]+\frac{\mu+\tilde{\lambda}}{\mu}\mathbb{E}\left[N^{s}1_{A=0}\right]-\sum_{j=1}^{s}\binom{s}{j}(-1)^j\left(\mathbb{E}\left[N^{s-j+1}\right]-\mathbb{E}\left[N^{s-j}1_{A=0}\right]\right).\\
\end{split}
\end{equation}	
\end{theorem}
These moments exist for all $s\in\mathbb{N}$, as in both states $N$ is dominated by the Poisson distribution with parameter $\max(\lambda,\tilde\lambda+\mu)/\mu$.

\begin{corollary}[The second moment of the protein number distribution]
The second moments of the protein distribution in the steady state are characterised by
  \label{C2} 
\begin{equation*} \mathbb{E}\left[1_{A=1}N^2\right]=\frac{1}{\psi}\left(P\left\{A=0\right\}\theta\left(1+\frac{\tilde{\lambda}}{\mu}-\frac{\theta}{\psi}-\frac{\mu}{\psi}\right)+\lambda P\left\{A=1\right\}\left(1+\frac{\theta}{\mu}\right)\right),
\end{equation*}
\begin{equation*} \mathbb{E}\left[N^2\right]=P\left\{A=0\right\}\left(\frac{\theta\left(\lambda-\tilde{\lambda}\right)}{\psi\mu}+1+\frac{\tilde{\lambda}}{\mu}\left(3+\frac{\tilde{\lambda}}{\mu}\right)\right)+\frac{\lambda}{\mu}P\left\{A=1\right\}\left(2+\frac{\tilde{\lambda}}{\mu}\right).
\end{equation*}
\end{corollary}

Theorem \ref{T2} and the Corollaries \ref{C1} and \ref{C2} are derived from the master equation \eqref{Master} and provide explicit formulas for moments of $X$ linear dependent on $P\left\{A=1\right\}$. However an explicit formula for $P\left\{A=1\right\}$ is not given. The following Theorem \ref{T3} can be used to compute $P\left\{A=1\right\}$ and more general the distribution $X$ numerically.

\begin{theorem} [Relation between $X$ and the recursion \eqref{E8}]
 \label{T3} 
For $n\geq 1$ set
\begin{equation}
\begin{split}\label{Rec1}
x_{n+2}=&\frac{1}{n(n+2)\mu^2}\left[x_{n+1}\mu\left(n((n+1)(\psi+\mu)+\lambda)+(\tilde{\lambda}+\theta)(n+1)\right)\right.\\
&\left.-x_{n}\left(\tilde{\lambda}\left(n(\psi+\mu)+\lambda\right)+n\lambda\mu+\lambda\theta\right)+
x_{n-1}\lambda\tilde{\lambda}\right].
\end{split}
\end{equation}
For each combination of $\psi$, $\theta$, $\mu$, $\lambda\in\mathbb{R}^{>0}$, $\tilde{\lambda}\in\mathbb{R}^{\geq 0}$, $n\geq 1$ and $x_{n-1}>0$ (respectively $x_{n+1}>0$), there is only one combination of $x_{n},x_{n+1}>0$ (respectively $x_{n},x_{n-1}>0$), so that the recursively determined sequence defined by \eqref{Rec1} determined by the values $x_{n-1}$, $x_{n}$ and $x_{n+1}$ at position $n-1,n$ and $n+1$ respectively has the features of a measure.
\end{theorem}
Due to the linearity of \eqref{Rec1} (respectively \eqref{E8}, which is \eqref{Rec1}  evaluated with the probabilities of the unbound state) in $x_{n+1}$, $x_n$ and $x_{n-1}$ this sequence is proportional to the sequence $\left(P\left\{X=(n,1)\right\}\right)_{n\in\mathbb{N}}$. It is straightforward to compute a sequence proportional (with the same factor) to $\left(P\left\{X=(n,0)\right\}\right)_{n\in\mathbb{N}}$ with the equation \eqref{E3}. Given both sequences the probability distribution of $X$ can be obtained by normalisation.\newline
Furthermore using Theorem \ref{T3} the interval, in which $P\left\{A=1\right\}$ lies can be estimated.

\begin{corollary}[Estimation of $P\left\{A=1\right\}$]
\label{R0}
For $\tilde\lambda=\lambda$ the probability distributions of the protein number in both states are Poisson and it holds \begin{equation}\label{eq_lambda} P\left\{A=1\right\}=\frac{\mu\theta}{\mu\theta+\lambda\psi}.\end{equation}
Generally it holds \begin{equation}\label{ineq}\frac{\mu\theta}{\mu\theta+\max\left(\lambda,\tilde\lambda\right)\psi}\leq P\left\{A=1\right\}\leq\frac{\mu\theta}{\mu\theta+\min\left(\lambda,\tilde\lambda\right)\psi}.\end{equation}
\end{corollary}
Using Theorem \ref{T2} more estimations can be made, as for all $s\in\mathbb{N}$ the terms $\mathbb{E}\left[N^s\right],\mathbb{E}\left[N^s1_{A=1}\right]$ and $\mathbb{E}\left[N^s1_{A=0}\right]$ can be expressed as linear terms in $P\left\{A=1\right\}$. So claims like $\mathbb{E}\left[N^{s+1}1_{A=0}\right]\geq \mathbb{E}\left[N^s1_{A=0}\right]$ lead to estimations for $P\left\{A=1\right\}$. However these estimations get complexer with increasing $s$ and are in general not better than \eqref{ineq}.\newline
For example considering $\mathbb{E}\left[N1_{A=0}\right]\geq P\left\{A=0\right\}$ it can be derived by using \eqref{E6} for $\lambda\psi+\theta\mu\geq\tilde\lambda\psi$, that \begin{equation*}
\begin{split}
\frac{\mu\theta-\tilde\lambda\psi}{\psi\left(\lambda-\tilde\lambda\right)+\mu\theta}\leq P\left\{A=1\right\}.
\end{split}
\end{equation*}
For $\tilde\lambda=0$ this corresponds to \eqref{ineq}. Only if additionally $\tilde\lambda>\lambda$ and $\tilde\lambda^2-\tilde\lambda\frac{\mu\theta}{\psi}+\mu\lambda<0$, this estimation is better than \eqref{ineq}.\newline
The previous part provides a characterisation of $X$. Theorem \ref{T3} provides the basis for an algorithm presented in Section \ref{section_algo}, which can compute the probability distribution of $X$. Next the distributions of proteins at the end of the unbound and bound state are considered.

\begin{definition}
\label{D1}
Let $C_1$ be the protein number at the end of the bound state and $B_1$ the protein number at the end of the unbound state. \newline
Let $C_0$ (respectively $B_0$) be the protein number at the beginning of a unbound (respectively bound) state, which ends with $B_1$ (respectively $C_1$) proteins. Let $C$ and $B$ be the corresponding steady state distributions.
\end{definition}
It is clear, that $P\left\{C_1=0\right\}=P\left\{B_1=0\right\}=0$.

\begin{theorem} [the bound state]
\label{T4} 
The expected number of proteins at the end of a bound state, starting with $b\geq 1$ proteins is given by
\begin{equation}
\begin{split}\label{B_eq1}
\mathbb{E}\left[{C_1|B_0=b}\right]=
%\left(b-\frac{\tilde{\lambda}}{\mu}-1\right)\frac{\theta}{\theta+\mu}+1+\frac{\tilde{\lambda}}{\mu }=
\frac{b\theta+\tilde\lambda+\mu}{\theta+\mu}.
\end{split}
\end{equation}
There is a linear dependence between the expected number of $C$ and $B$ in the equilibrium
\begin{equation}
\begin{split}\label{B_eq2}
\mathbb{E}\left[C\right]=\frac{\mathbb{E}\left[B\right]\theta+\tilde\lambda+\mu}{\theta+\mu}.
%\left(\mathbb{E}\left[B\right]-\frac{\tilde{\lambda}}{\mu}-1\right)\frac{\theta}{\theta+\mu}+1+\frac{\tilde{\lambda}}{\mu }.
\end{split}
\end{equation}
For $c,b\geq1$ the conditioned probability distribution of $C_1$ given $B_0=b$ is
\begin{equation*}
\begin{split}
&P\left\{C_1=c|B_0=b\right\}=\\
&\sum_{n=\max(0,b-c)}^{b-1}\binom{b-1}{n}\frac{\theta\left(\frac{\tilde{\lambda}}{\mu}\right)^{c-b+n}}{(c-b+n)!}\int_{0}^{\infty}e^{-t\left(\theta+n\mu\right)-\frac{\tilde{\lambda}}{\mu}(1-e^{-t\mu})}(1-e^{-t\mu})^{c-1}dt.\\
\end{split}
\end{equation*}
\end{theorem}
If the case $\tilde{\lambda}=0$, where there is no production in the bound state, is considered, $P\left\{C_1=c|B_0=b\right\}$ can be displayed explicitly for $c\geq 1$
\begin{equation*}
\begin{split}
P\left\{C_1=c|B_0=b\right\}=\theta\binom{b-1}{c-1}\sum_{s=0}^{c-1}\binom{c-1}{s}\frac{(-1)^s}{\theta+\mu(b-1-s)}.
\end{split}
\end{equation*}
Considering \eqref{B_eq2} it can be seen, that if $\lambda$ is changed while all other parameters are fixed, $\mathbb{E}\left[C\right]$ and $\mathbb{E}\left[B\right]$ are linearly dependent.

\begin{theorem} [The distribution $C$]
\label{T5} 
The distribution of proteins at the end of the bound state in the steady state equals the distribution of proteins during the bound state.
\end{theorem}

\begin{theorem} [The distribution $B$]
\label{T6} 
The first two moments of $B$ in the steady state are 
\begin{equation*}
\begin{split}
\mathbb{E}\left[B\right]&=\frac{\mathbb{E}\left[C\right]\left(\theta+\mu\right)-\tilde{\lambda}-\mu}{\theta},\\
\mathbb{E}\left[B^2\right]&=\frac{\mathbb{E}\left[C^2\right]\left(\theta+2\mu\right)}{\theta}-\frac{\mu\mathbb{E}\left[B\right]}{\theta+\mu}\left(3\mu \theta+\frac{2\tilde{\lambda}}{\mu}\right)-\left(2+\frac{2\tilde{\lambda}}{\mu}+\left(\frac{\tilde{\lambda}}{\mu}\right)^2\right)\\
&+\frac{\theta+2\mu}{\theta+\mu}\left(3+\frac{5\tilde{\lambda}}{\mu}+2\left(\frac{\tilde{\lambda}}{\mu}\right)^2\right)-(\theta+2\mu)\theta\left(1+3\frac{\tilde{\lambda}}{\mu}+\left(\frac{\tilde{\lambda}}{\mu}\right)^2\right).
\end{split}
\end{equation*}
\end{theorem}
If $\lambda=\tilde\lambda$, it follows directly with Corollary \ref{R0}, that $\mathbb{E}\left[C\right]=\mathbb{E}\left[B\right]=1+\frac{\lambda}{\mu}$.
\begin{remark}[computation of $B$]
\label{R1}
Using Markov chain theory it is possible but numerically costly to compute $P\left\{B_1=b|C_0=c\right\}$ (or $P\left\{C_1=c|B_0=b\right\}$). For this purpose the transition probabilities of getting from $C_0=c$ to $B_1=b$ in $m\in\mathbb{N}$ reaction steps have to be computed. Hence the distribution of $B$ in the steady state can be computed exactly
with some numerical effort given the distribution $P\left\{X=(i,0)\right\}$ for $i\in\mathbb{N}$, as by Theorem \ref{T5}
\begin{equation*}
\begin{split}
P\left\{B=b\right\}=\sum_{c=0}^{\infty}P\left\{B_1=b|C_0=c\right\}P\left\{X=(c,0)\right\}.
\end{split}
\end{equation*} 
\end{remark}
Next the case $\lambda=0$ is considered. Here it is certain, that during some unbound state all protein decay. Hence the steady state only consists of the absorbing state with no protein, so $P\left\{X=(0,1)\right\}=1$.

\begin{definition}
  \label{D2} 
Let $\lambda=0$. Consider $X_t$ starting at $t=0$. Let $S$ be the number of bound-unbound cycles until extinction and $T_i$ be the length of time of $i$ bound-unbound-cycles.
\end{definition}

\begin{theorem} [the case $\lambda=0$]
  \label{T7} 
Let $\lambda=0$ and $i\geq 2$, then conditioned on the protein number $n\geq 1$ at time $0$, the probability distribution of $S$ is given by
\begin{equation*}
\begin{split}
&P\left\{S=1|X_0=(n,0)\right\}=\sum_{m=1}^{\infty}P\left\{C_1=m|B_0=n\right\}\left(\frac{\mu}{\mu+\psi}\right)^{m},\end{split}
\end{equation*}
\begin{equation*}
\begin{split}
&P\left\{S=i|X_0=(n,0)\right\}=\\
&\sum_{m=1}^{\infty}\frac{\psi P\left\{C_1=m|B_0=n\right\}}{\mu+\psi}\sum_{j=1}^{m}\left(\frac{\mu}{\mu+\psi}\right)^{m-j}P\left\{S=i-1|X_0=(j,0)\right\}.\end{split}
\end{equation*}
Furthermore it holds
\begin{equation*}
\begin{split}
&\mathbb{E}\left[T_1|X_0=(n,0),S>1\right]=\\
&\frac{1}{\theta}+\sum_{m=1}^{\infty}\frac{\psi P\left\{C_1=m|B_0=n\right\}}{\mu+\psi}\sum_{j=1}^{m}\left(\frac{\mu}{\mu+\psi}\right)^{m-j}\left[\frac{1}{\mu}\sum_{i=j+1}^{m}\frac{1}{i}+\frac{1}{\psi j}\right],\end{split}
\end{equation*}
\begin{equation*}
\begin{split}
&\mathbb{E}\left[T_i|X_0=(n,0),S>i\right]=\frac{1}{\theta}+\sum_{m=1}^{\infty}\frac{\psi P\left\{C_1=m|B_0=n\right\}}{\mu+\psi}\sum_{j=1}^{m}\left(\frac{\mu}{\mu+\psi}\right)^{m-j}\\
&\qquad\qquad\qquad\qquad\left(\frac{1}{\mu}\sum_{k=j+1}^{m}\frac{1}{k}+\frac{1}{\psi j}+\mathbb{E}\left[T_{i-1}|X_0=(j,0),S>i-1\right]\right).\end{split}
\end{equation*}
For the first two moments of the number of bound-unbound cycles until extinction starting with $n\geq 1$ the following fixed-point equations hold:
\begin{equation*}
\begin{split} 
&\mathbb{E}\left[S|X_0=(n,0)\right]\\
&=\sum_{m=1}^{\infty}\frac{\psi P\left\{C_1=m|B_0=n\right\}}{\mu+\psi}\sum_{j=1}^{m}\left(\frac{\mu}{\mu+\psi}\right)^{m-j}\left(\mathbb{E}\left[S|X_0=(j,0)\right]+1\right),
\end{split}
\end{equation*}
\begin{equation*}
\begin{split}
&\mathbb{E}\left[S^2|X_0=(n,0)\right]=2\mathbb{E}\left[S|X_0=(n,0)\right]\\
&+\sum_{m=1}^{\infty}P\left\{C_1=m|B_0=n\right\}\sum_{j=1}^{m}\left(\frac{\mu}{\mu+\psi}\right)^{m-j}\frac{\psi\left(\mathbb{E}\left[S^2|X_0=(j,0)\right]-1\right)}{\mu+\psi}.
\end{split}
\end{equation*}
If the unbound state starting with $n$ proteins ends after a finite time, it holds:
\begin{equation*}
\begin{split}
\mathbb{E}\left[B_1|C_0=n\right]&=\frac{\mu+\psi}{\psi}\left(1-\left(\frac{\mu}{\mu+\psi}\right)^{n}\left(1+\frac{n\psi}{\mu+\psi}\right)\right).
\end{split}
\end{equation*}
\end{theorem}
Let $i,n\in\mathbb{N}$. For the computation of $P\left\{S=i|X_0=(n,0)\right\}$, $\mathbb{E}\left[S|X_0=(n,0)\right]$, $\mathbb{E}\left[S^2|X_0=(n,0)\right]$ and $\mathbb{E}\left[T_i|X_0=(n,0),i<S\right]$ the conditional probabilities given in Remark \ref{R1} are needed. Hence their numerical computation is costly.

\begin{corollary}
  \label{C4} 
Let $\lambda=0$. Let $C_{1}$ be the distributions of proteins at the beginning of a unbound state and $C_2$ be the corresponding distribution at the beginning of the following unbound state. Given the first unbound state ends after a finite time and starts with $m$ proteins, it holds for $m\geq 1$
\begin{equation*}
\begin{split}
&\mathbb{E}\left[C_2|C_1=m\right]=\frac{\tilde\lambda+\mu}{\theta+\mu}-\frac{\theta}{\theta+\mu}\frac{\mu+\psi}{\psi}\left(\left(\frac{\mu}{\mu+\psi}\right)^m\left(m\frac{\psi}{\mu+\psi}+1\right)-1\right),\\
&\frac{\tilde\lambda+\mu}{\theta+\mu}+\frac{\theta}{\theta+\mu}\frac{\psi}{\mu+\psi} \leq \mathbb{E}\left[C_2|C_1=m\right] \leq \frac{\tilde\lambda+\mu}{\theta+\mu}+\frac{\theta}{\theta+\mu}\frac{\mu+\psi}{\psi}.
\end{split}
\end{equation*}
Given $i$ bound-unbound cycles, the distribution of totally produced proteins can be calculated: $$P\left\{m\mbox{  produced proteins}|i \mbox{ cycles }\right\}=\left(\frac{\theta}{\theta+\tilde{\lambda}}\right)^{i}\left(\frac{\tilde{\lambda}}{\theta+\tilde{\lambda}}\right)^{m} \binom{m+i-1}{i}.$$
\end{corollary} 

\subsection{Algorithm for the computation of the distribution $X$}\label{section_algo}
Next a method to compute the equilibrium distribution of protein numbers in the bound and unbound state is presented. 

\begin{definition}[the recursion $R$]
 \label{D3} 
For $i,n\in\mathbb{N}$ and $x,y,z\in\mathbb{R}$ let $R_{i,n}(x,y,z)$ be the recursion \eqref{Rec1}  
determined by fixed $x,y,z$ at position $n-1,n,n+1$ evaluated at position $i$.
\end{definition}
Hence given $n\in\mathbb{N}$ it holds $$R_{n-1,n}(x,y,z)=x, R_{n,n}(x,y,z)=y, R_{n+1,n}(x,y,z)=z,$$ for all $i\geq n$
\begin{equation}
\begin{split}\label{Rec_bsp}
&R_{i+2,n}(x,y,z)=\frac{1}{i(i+2)\mu^2}\left[-R_{i,n}(x,y,z)\left(\tilde{\lambda}\left(i(\psi+\mu)+\lambda\right)+i\lambda\mu+\lambda\theta\right)\right.\\
&\left.+R_{i+1,n}(x,y,z)\mu\left(i((i+1)(\psi+\mu)+\lambda)+(\tilde{\lambda}+\theta)(i+1)\right)+
R_{i-1,n}(x,y,z)\lambda\tilde{\lambda}\right]
\end{split}
\end{equation}
and for all $1\leq i\leq n-1$ (and $\tilde\lambda\neq0$)
\begin{equation}
\begin{split}\label{Rec_bsp2}
&R_{i-1,n}(x,y,z)=\frac{1}{\lambda\tilde{\lambda}}\left[-R_{i+1,n}(x,y,z)\mu\left(i((i+1)(\psi+\mu)+\lambda)+(\tilde{\lambda}+\theta)(i+1)\right)\right.\\
&\left.+i(i+2)\mu^2R_{i+2,n}(x,y,z)+R_{i,n}(x,y,z)\left(\tilde{\lambda}\left(i(\psi+\mu)+\lambda\right)+i\lambda\mu+\lambda\theta\right)
\right].
\end{split}
\end{equation}
The goal is to find $\hat{x},\hat{y},\hat{z}$, so that $P\left\{X=(i,1)\right\}=R_{i,n}\left(\hat{x},\hat{y},\hat{z}\right)$ holds for all $i\in\mathbb{N}$. In the following the cases $n=1$ and $n\geq 2$ are distinguished. By Theorem \ref{T3} the solution is unique.
Due to the linearity of \eqref{Rec1} it holds for $x\neq 0$, that
\begin{equation}
\begin{split}\label{lin}
&R_{i,n}\left(x,y,z\right)=xR_{i,n}\left(1,\frac{y}{x},\frac{z}{x}\right).
\end{split}
\end{equation}
So if the starting value at position $n-1$ is set to one, by Theorem \ref{T3} there are unique $(\tilde{y},\tilde{z})$, so that the sequence $\left(R_{i,n}(1,\tilde{y},\tilde{z})\right)_{i\in\mathbb{N}}$ has the features of a measure.\newline
Let $Y$ be Poisson distributed with parameter $\max\left(\frac{\lambda}{\mu},\frac{\tilde{\lambda}}{\mu}+1\right)$. It is easy to see that there exists a $m\in\mathbb{N}$, so that $P\left\{X=(n,1)\right\}\leq P\left\{Y=n\right\}$ for all $n\geq m$. Hence for $n$ big enough it holds $P\left\{X=(n,1)\right\}\approx 0$.

\subsubsection{starting the algorithm at position zero}
If the recursion is started at the beginning, choose $\tilde{x}=1$ and by \eqref{Startpaar} $\tilde{y}=\frac{\lambda}{\mu}$. So only the unique $\tilde{z}$, for which the sequence $\left(R_{i,1}\left(1,\frac{\lambda}{\mu},\tilde{z}\right)\right)_{i\in\mathbb{N}}$ has the features of a measure, needs to be determined. It holds $\tilde{z}=\frac{P\left\{X=(2,1)\right\}}{P\left\{X=(0,1)\right\}}$ by \eqref{lin}.
Furthermore due to the linearity of the recursion there are $q_i, r_i\in\mathbb{R}$ for each $i\in\mathbb{N}$, with $R_{i,1}(1,\frac{\lambda}{\mu},z)=q_i z+r_i$.
Given two unequal points $z_{1},z_{2}$, it is straightforward to compute for $i\in\mathbb{N}$
\begin{equation}
\begin{split}\label{q,r}
&q_i=\frac{R_{i,1}\left(1,\frac{\lambda}{\mu},z_{1}\right)-R_{i,1}\left(1,\frac{\lambda}{\mu},z_{2}\right)}{z_{1}-z_{2}},\\
&r_i=R_{i,1}\left(1,\frac{\lambda}{\mu},z_{1}\right)-q_i z_{1}.\\
\end{split}
\end{equation}
Let $m$ be large enough (as just described), it can be expected that the unique $\tilde{z}$ fulfills $R_{m,1}\left(1,\frac{\lambda}{\mu},\tilde{z}\right)=q_m\tilde{z}+r_m\approx 0$. Hence given two unequal points $z_{1},z_{2}$, the following approximation $z'$ for $\tilde{z}$ holds
\begin{equation}
\begin{split}\label{xhat2}
z'=-\frac{r_m}{q_m}=z_1-R_{m,1}\left(1,\frac{\lambda}{\mu},z_{1}\right)\frac{z_1-z_2}{R_{m,1}\left(1,\frac{\lambda}{\mu},z_{1}\right)-R_{m,1}\left(1,\frac{\lambda}{\mu},z_{2}\right)}.
\end{split}
\end{equation}
Due to numerical inaccuracy it is advisable to compute $z'$ more than once with different starting values $z_1$ and $z_2$.
Considering the recursion, it can be shown that for all $z\neq \tilde{z}$ the sequence $\left(|R_{i,1}(1,\frac{\lambda}{\mu},z)|\right)_{i\in\mathbb{N}}$ is unbounded. \newline
The corresponding sequence of the bound state can be computed using \eqref{E1} and \eqref{E4}. To compute the corresponding probabilities, both sequences are scaled with their total sum.\newline
Now  the algorithm can be outlined for $n=1$:
\begin{enumerate}
	\item Choose $m>>n$ so that $\frac{\max(\lambda,\tilde\lambda+\mu)^m}{\mu^m m!}e^{-\frac{\max(\lambda,\tilde\lambda+\mu)}{\mu}}\approx 0$
	\item Choose $z_1,z_2>0$ with $z_1\neq z_2$
	\item Compute $q_m,r_m$ using \eqref{q,r} 
	\item Compute $z'$ using \eqref{xhat2}
	\item Compute the sequence $\left(R_{i,1}\left(1,\frac{\lambda}{\mu},z'\right)\right)_{0\leq i\leq m}$ using the recursion \eqref{Rec_bsp}
	\item Compute the bound state-sequence $\left(S_i\right)_{0\leq i\leq m}$ using the recursion \eqref{E4}, \eqref{E9} and the sequence $\left(R_{i,1}\left(1,\frac{\lambda}{\mu},z'\right)\right)_{0\leq i\leq m}$
	\item Set for all $0\leq i \leq m$ 
\begin{equation*}
\begin{split}
&P\left\{X=(i,1)\right\}=\frac{R_{i,1}\left(1,\frac{\lambda}{\mu},z'\right)}{\sum_{j=0}^{m}R_{j,1}\left(1,\frac{\lambda}{\mu},z'\right)+S_j},\\
&P\left\{X=(i,0)\right\}=\frac{S_i}{\sum_{j=0}^{m}R_{j,1}\left(1,\frac{\lambda}{\mu},z'\right)+S_j}.
\end{split}
\end{equation*}
\end{enumerate}
Computation of the protein distribution can lead to numerical problems depending on the parameters $\psi,\lambda,\tilde{\lambda},\mu$ and $\theta$. If $\frac{\lambda}{\mu}$ is large, the probabilities $P\left\{X=(1,1)\right\}$ and $P\left\{X=(2,1)\right\}$ may lie near the machine precision. In such cases it may be beneficial to either start the recursion at a position $n\geq 2$, or if $\tilde\lambda$ is relatively small to use an analogous algorithm based on the recursion for the bound state.

\subsubsection{starting the algorithm in $n\geq 2$}
In some cases it is beneficial to start the recursion at a position $n\geq 2$ and to compute the $P\left\{X=(i,1)\right\}$ for all positions $i\in\mathbb{N}$, with $i<n-1$ and $i>n+1$. This approach works only for $\tilde\lambda>0$. In contrast to the case where the recursion is started at the beginning, two instead of one parameters have to be determined. However again the linearity of the recursion can be used. Analogously to the previous case a $m>>n$ is choosen, for which the corresponding probability $P\left\{X=(m,1)\right\}$ is expected to be nearly zero. Set the parameter $\tilde{x}=1$. Additionally condition \eqref{Startpaar} should hold for the backward recursion, hence $\tilde{y},\tilde{z}$ are searched, who fulfill the conditions 
\begin{equation}
\begin{split}\label{n>1conditions}
&\lambda R_{0,n}(1,\tilde{y},\tilde{z})=\mu R_{1,n}(1,\tilde{y},\tilde{z}),\\
&R_{m,n}(1,\tilde{y},\tilde{z})\approx 0.
\end{split}
\end{equation}
Consider a fixed $n\in\mathbb{N}$. For each $i\in\mathbb{N}$, there are $a_i, b_i, c_i\in\mathbb{R}$ with $R_{i,n}(1,y,z)=a_i y+b_i z+c_i$.
Given $y_1,y_2,z_{1},z_{2}$ with $y_1\neq y_2$ and $z_{1}\neq z_{2}$ it is straightforward to compute for $i\in\mathbb{N}$
\begin{equation}
\begin{split}\label{a,b,c}
&a_i=\frac{R_{i,n}(1,y_1,z_{1})-R_{i,n}(1,y_2,z_{1})}{y_1-y_2},\\
&b_i=\frac{R_{i,n}(1,y_1,z_{1})-R_{i,n}(1,y_1,z_{2})}{z_{1}-z_{2}},\\
&c_i=-a_iy_1-b_iz_{1}+R_{i,n}(1,y_1,z_{1}).\\
\end{split}
\end{equation}
By using the conditions \eqref{n>1conditions} an approximation $(y',z')$ for $(\tilde{y},\tilde{z})$
\begin{equation}
\begin{split}\label{xhats}
&z'=\frac{a_m\tilde{c}-c_m\tilde{a}}{b_m\tilde{a}-a_m\tilde{b}},\\
&y'=-\frac{c_m}{a_m}-\frac{b_m}{a_m}\frac{a_m\tilde{c}-c_m\tilde{a}}{b_m\tilde{a}-a_m\tilde{b}}
\end{split}
\end{equation}
is obtained with
\begin{equation}
\begin{split}\label{tilde,a,b,c}
\tilde{a}:=\frac{\lambda}{\mu}a_0-a_1,\tilde{b}:=\frac{\lambda}{\mu}b_0-b_1\mbox{ and } \tilde{c}:=\frac{\lambda}{\mu}c_0-c_1.
\end{split}
\end{equation}
Analogously we need to compute the distribution of the bound state and to normalize the solution. However due to numerical reasons, the choice of $y_1$, $y_2$, $z_{1}$, $z_{2}$ and $n$ is important and should be considered. \newline
Now the algorithm for $n\geq 2$ and $\tilde\lambda>0$ can be outlined:
\begin{enumerate}
	\item Choose $m>>n$ so that $\frac{\max(\lambda,\tilde\lambda+\mu)^m}{\mu^m m!}e^{-\frac{\max(\lambda,\tilde\lambda+\mu)}{\mu}}\approx 0$
	\item Choose $y_1,z_{1},y_2,z_{2}>0$ with $y_1\neq y_2$ and $z_{1}\neq z_{2}$
	\item Compute $a_m,a_1,a_0,b_m,b_1,b_0,c_m,c_1,c_0$ using \eqref{a,b,c} and $\tilde{a},\tilde{b},\tilde{c}$ using \eqref{tilde,a,b,c}
	\item Compute $y',z'$ using \eqref{xhats}
	\item Compute the sequence $\left(R_{i,n}(1,y',z')\right)_{0\leq i\leq m}$ using the recursions \eqref{Rec_bsp} and \eqref{Rec_bsp2}
	\item Compute the bound state-sequence $\left(S_i\right)_{0\leq i\leq m}$ using the recursion \eqref{E4}, \eqref{E9} and the sequence $\left(R_{i,n}(1,y',z')\right)_{0\leq i\leq m}$
	\item Set for all $0\leq i \leq m$ 
\begin{equation*}
\begin{split}
&P\left\{X=(i,1)\right\}=\frac{R_{i,n}(1,y',z')}{\sum_{j=0}^{m}R_{j,n}(1,y',z')+S_j},\\
&P\left\{X=(i,0)\right\}=\frac{S_i}{\sum_{j=0}^{m}R_{j,n}(1,y',z')+S_j}.
\end{split}
\end{equation*}
\end{enumerate}

\subsection{The impact of the feedback on the system}\label{feedback_impact}
In this section the impact of feedback on the system is analysed. Therefore a non-feedback model (NFM) describing gene expression is introduced. To compare this model with a model describing feedback, the feedback-model described in Section \ref{section_Main_results} is modified slightly. This feedback-model is called modified feedback model (MFM).

\subsubsection{the non-feedback model (NFM)}
The NFM equals the gene expression model except for two points. First all protein can decay during the bound state and second the unbound states length is exponentially distributed with a parameter $\tilde\psi$. Thus its duration does not depend on the current number of proteins. The NFM reactions are
\begin{align*}
  1: && U & \xrightarrow{\lambda}U+P,\\
  2: && B & \xrightarrow{\tilde\lambda}B+P,\\
  3: && U & \xrightarrow{\tilde\psi} B,\\
  4: && B & \xrightarrow{\theta} U,\\
  5: && P & \xrightarrow{\mu} \emptyset.
\end{align*}
Hence in contrast to the model described in Section \ref{section_Main_results}, only first order reactions are considered, which simplifies the model alot. \cite{peccoud1995markovian} introduced a special case of this model, where $\tilde\lambda=0$.
Analogously to the feedback-model let $M$ and $\tilde{A}$ be the protein number and the state of the gene in the NFM in the steady state. It holds $P\left\{\tilde{A}=1\right\}=\frac{\theta}{\theta+\tilde\psi}$. Using the equation corresponding to \eqref{E2} and the methods of Theorem \ref{T2} it can be derived, that
\begin{equation}
\begin{split}\label{M_1st_moment}
\mathbb{E}\left[M\right]=\frac{\lambda}{\mu}P\left\{\tilde{A}=1\right\}+\frac{\tilde\lambda}{\mu}P\left\{\tilde{A}=0\right\},
\end{split}
\end{equation}
\begin{equation}
\begin{split}\label{M_2nd_moment}
\mathbb{E}\left[M^2\right]=\frac{\lambda}{\mu}\left(P\left\{\tilde{A}=1\right\}+\mathbb{E}\left[1_{\tilde{A}=1}M\right]\right)+\frac{\tilde\lambda}{\mu}\left(P\left\{\tilde{A}=0\right\}+\mathbb{E}\left[1_{\tilde{A}=0}M\right]\right).
\end{split}
\end{equation}
In contrast to the feedback-model all moments can be derived exactly. In the following the expected number of proteins in the unbound state is derived.\newline
Advancing similar to Theorem \ref{T4} and using the notation $B$ and $C$ of Definition \ref{D1} for the protein number in the steady state at the end of the unbound and bound state in the NFM, it can be proved, that
\begin{equation}
\begin{split}\label{EB}
\mathbb{E}\left[B\right]=\frac{\mathbb{E}\left[C\right]\tilde\psi+\lambda}{\mu+\tilde\psi},
\end{split}
\end{equation}
\begin{equation}
\begin{split}\label{EC}
\mathbb{E}\left[C\right]=\frac{\mathbb{E}\left[B\right]\theta+\tilde\lambda}{\mu+\theta}.
\end{split}
\end{equation}
Furthermore both states durations are exponentially distributed. So as seen in Theorem \ref{T5} it holds, $\mathbb{E}\left[M|\tilde{A}=1\right]\!=\mathbb{E}\left[B\right]$ and $\mathbb{E}\left[M|\tilde{A}=0\right]\!=\mathbb{E}\left[C\right]$. Hence it can be derived, that 
\begin{equation}
\begin{split}\label{EM}
\mathbb{E}\left[M|\tilde{A}=1\right]=\frac{\lambda(\theta+\mu)+\tilde\psi\tilde\lambda}{\mu(\mu+\tilde\psi+\theta)}.
\end{split}
\end{equation}

\subsubsection{the modified feedback model (MFM)}
To compare the NFM to a corresponding and suitable feedback-model, the MFM is constructed. Therefore the model described in Section \ref{section_Main_results} is changed in the following way: the protein bound to the gene can decay, which does not interrupt the bound states duration. The reactions are
\begin{align*}
  1: && U & \xrightarrow{\lambda}U+P,\\
  2: && B & \xrightarrow{\tilde\lambda}B+P,\\
  3: && U + P & \xrightarrow{\psi} B + P,\\
  4: && B & \xrightarrow{\theta} U,\\
  5: && P & \xrightarrow{\mu} \emptyset.
\end{align*}
Let $\bar{N}$ and $\bar{A}$ be the protein number and the state of the gene in the steady state. Analogously to the methods used in this paper (compare to Corollary \ref{C1} and Theorem \ref{T2}) it can be derived, that 
\begin{equation}
\begin{split}\label{ENA1}
\mathbb{E}\left[\bar{N}|\bar{A}=1\right]=\frac{P\left\{\bar{A}=0\right\}}{P\left\{\bar{A}=1\right\}}\frac{\theta}{\psi},
\end{split}
\end{equation}
\begin{equation}
\begin{split}\label{N_1st_moment}
\mathbb{E}\left[\bar{N}\right]=\frac{\lambda}{\mu}P\left\{\bar{A}=1\right\}+\frac{\tilde\lambda}{\mu}P\left\{\bar{A}=0\right\},
\end{split}
\end{equation}
\begin{equation}
\begin{split}\label{N_2nd_moment}
\mathbb{E}\left[\bar{N}^2\right]=\frac{\lambda}{\mu}\left(P\left\{\bar{A}=1\right\}+\mathbb{E}\left[1_{\bar{A}=1}\bar{N}\right]\right)+\frac{\tilde\lambda}{\mu}\left(P\left\{\bar{A}=0\right\}+\mathbb{E}\left[1_{\bar{A}=0}\bar{N}\right]\right).
\end{split}
\end{equation}

\subsubsection{Comparism of both models}
To compare $M$ and $\bar{N}$, fix $\lambda,\mu,\theta$ and $\psi$ and set 
\begin{equation}
\begin{split}\label{psitilde}
\tilde{\psi}=\theta\frac{P\left\{\bar{A}=0\right\}}{P\left\{\bar{A}=1\right\}}.
\end{split}
\end{equation} 
As consequence it follows $P\left\{\tilde{A}=1\right\}=P\left\{\bar{A}=1\right\}$ and $\mathbb{E}\left[M\right]=\mathbb{E}\left[\bar{N}\right]$, so in the steady state both models have the same protein in- and efflux.
\begin{proposition}[Relation between the variance and the expectation conditional on the state of the gene between the MFM and the NFM]\label{Comparism_Prop}
If \eqref{psitilde} holds, it follows that 
\begin{equation}
\begin{split}\label{Comparism_exp}
\mathbb{E}\left[\bar{N}|\bar{A}=1\right]<\mathbb{E}\left[M|\tilde{A}=1\right],
\end{split}
\end{equation}
\begin{equation}
\begin{split}\label{Comparism_var}
&V\left(\bar{N}\right)\leq V\left(M\right)\mbox{, if }\lambda\geq\tilde\lambda\mbox{ and}\\
&V\left(\bar{N}\right)\geq V\left(M\right)\mbox{, if }\lambda\leq\tilde\lambda.
\end{split}
\end{equation}
\end{proposition}

\subsection{Impact of stochasticity on the model}
\cite{hornos2005self} discuss the impact of stochasticity on the system of the self-regulating gene, by comparing the features of the model to the corresponding features of the common deterministic mass-action approach \citep{ackers1982quantitative}, which neglects the promoter fluctuation, hence the discrete distinction between the bound and unbound state. Here the relevance of stochasticity is illustrated by estimating the maximal difference between the model and the deterministic mass-action approach, if $\tilde\lambda=0$ is fixed.\newline
Corresponding to $N$ in the stochastic case, let $\tilde{N}$ be the distribution of protein-number in the steady state in the deterministic approach.
In the steady state protein-production equals decay of protein, so it holds
\begin{equation*}
\begin{split}
0=\frac{\theta}{\theta+\mathbb{E}\left[\tilde{N}\right]\psi}\left(\lambda-\mu\mathbb{E}\left[\tilde{N}\right]\right)+\frac{\mathbb{E}\left[\tilde{N}\right]\psi}{\theta+\mathbb{E}\left[\tilde{N}\right]\psi}\left(\tilde\lambda-\mu\left(\mathbb{E}\left[\tilde{N}\right]-1\right)\right).\\
\end{split}
\end{equation*}
Thus it holds that
\begin{equation*}
\begin{split}
\mathbb{E}\left[\tilde{N}\right]=\frac{1}{2\mu\psi}\left(\tilde\lambda\psi+\mu\left(\psi-\theta\right)+\sqrt{\left(\tilde\lambda\psi+\mu\left(\psi-\theta\right)\right)^2+4\lambda\theta\mu\psi}\right).
\end{split}
\end{equation*}
\begin{proposition}[Difference of the expected steady state protein number between the deterministic and the stochastic approach]\label{impact_Prop}
Let $c:=\frac{\lambda}{\mu}$. Consider $\tilde\lambda=0$, and fixed $\theta,\psi,c$. It holds $$\max_{\mu\in\mathbb{R}^{>0}}\left|\mathbb{E}\left[N\right]-\mathbb{E}\left[\tilde{N}\right]\right|\geq\frac{1}{2\psi}\left(\psi-\theta+\sqrt{(\psi-\theta)^2+4\theta\psi c}\right)-c\frac{\theta+\psi}{\theta+c\psi}.$$
$\mathbb{E}\left[N\right]$ can be smaller or greater than $\mathbb{E}\left[\tilde{N}\right]$.
\end{proposition}

\section{Discussion}\label{section_disc}
In this work the same model as described by \cite{hornos2005self} is considered. However additional to the analysis of the distribution of proteins in the unbound and bound state, the distributions at the end of these states and the case $\lambda=0$ are analysed. \newline
Furthermore a different method to calculate the protein distribution during the two states is presented. It uses only a recursive description of the probability distribution derived from the master equation.\newline
In this discussion the results are outlined and numerical examples are given. Next the here presented method to compute the protein-distribution is compared to the method of Hornos. Finally the impact of stochasticity on the model is discussed.\newline 
In Theorem \ref{T2} it is shown, that all higher moments of the bound and unbound state equal linear combinations of the expectation values of lower moments in the bound and unbound state. Hence given the correct $P\left\{A=1\right\}$ all moments can be calculated. However we have no formula for $P\left\{A=1\right\}$ given the reaction rates. In Section \ref{section_algo} an algorithm to compute $X$ and therewith $P\left\{A=1\right\}$ is presented.\newline 
It is shown, that the steady state protein distribution at the end of the bound state equals the distribution during the bound state. However the steady state protein distribution at the end of the unbound state is more complex to compute. Here the time in the unbound state depends on the evolution of protein numbers during the state. Only the first two moments are given in this paper. Higher moments are linearly dependent on moments of equal and lower order of $C$.
Figure \ref{pic:Fig1} displays the distribution of proteins at the end and during both states. In the bound state they are equal by Theorem \ref{T5}. In the unbound state the protein-mass is higher at the end of the state than during the state. The probability to have no proteins at the end of the state is zero.\newline
\begin{figure}[htbp]
\centering
\includegraphics[scale=0.8]{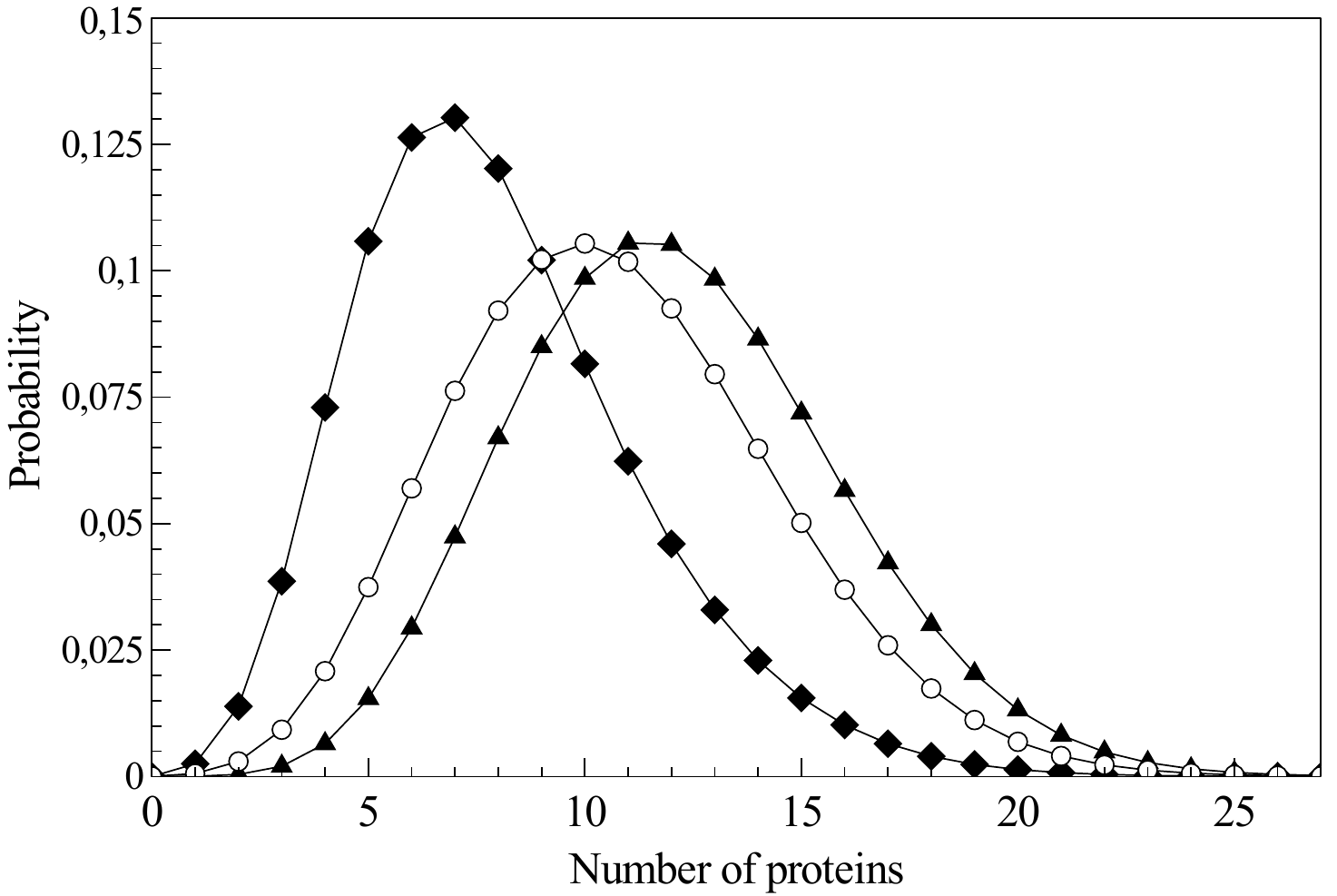}
\caption{Distribution of proteins at the beginning (diamonds), during (circles) and at the end of the unbound state (triangles) in the steady state. A positive feedback is considered. Parameter values are $\lambda=15,\tilde\lambda=5,\mu=1,\theta=0.5, \psi=0.1$. The distribution of proteins during the bound state equals the distribution at the beginning of the unbound state.}%
\label{pic:Fig1}
\end{figure}
Setting $\lambda=0$ leads to the fact, that the protein number in the steady state is $0$, as there is always the chance in the unbound state, that each protein decays, so no protein can initiate the bound state, in which proteins can be produced with rate $\tilde\lambda$. This special case is not common in gene regulation. However it gives some indication of how the system behaves for very small $\lambda$. The time and number of bound and unbound states, until the last protein decays, dependent on the initial number of proteins at time $0$ is considered. In Theorem \ref{T7} and Corollary \ref{C4} equations, which describe implicitly and explicitly certain parameters of this special positive feedback model, are derived. The probability $P\left\{C_1=m|B_0=n\right\}$ appearing in Theorem \ref{T7} can be numerically approached with Theorem \ref{T4}.\newline
\begin{figure}[htbp]
\centering
\includegraphics[scale=1.2]{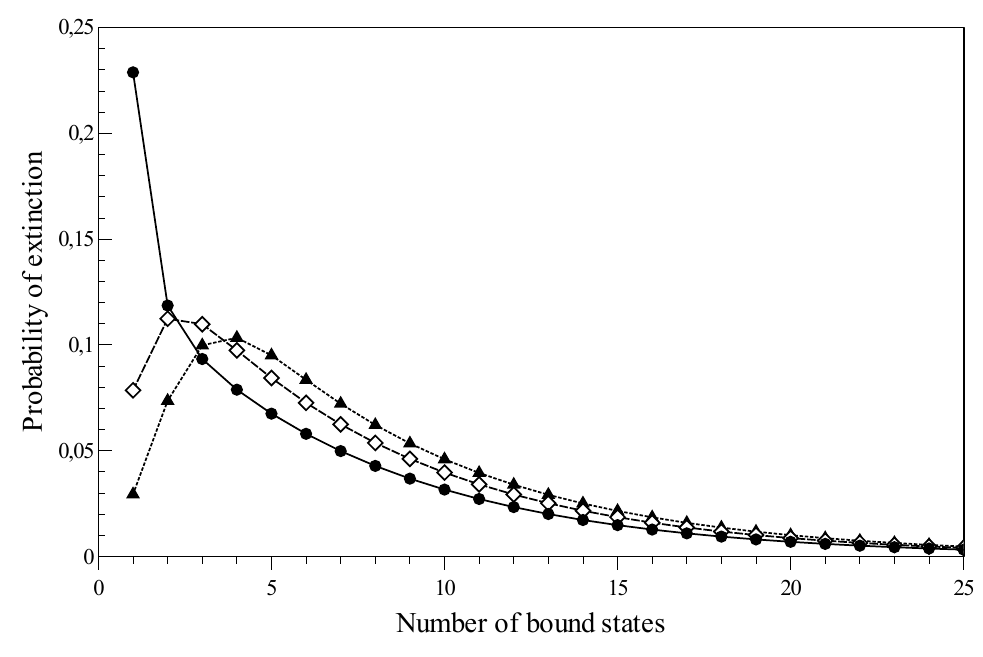}
\caption{Probability distribution of number of bound-unbound cycles until extinction of protein, given $n=1$ (solid line, circles), $16$ (dashed line, diamonds), $64$ (dotted line, triangles) proteins at the beginning of the first bound state. Parameter values are $\lambda=0,\tilde\lambda=10,\mu=1,\theta=0.5, \psi=0.25$ }
\label{pic:Fig3}
\end{figure}
Figure \ref{pic:Fig3} displays the probabilities of protein extinction after a corresponding number of bound states, depending on the number of proteins at time zero. In this example $\psi$ is relatively high, this causes a relatively short time in the unbound state and hence a relatively small probability of extinction.\newline
Figure \ref{pic:Fig2} displays the distribution of the number of proteins, in both states. This distribution was computed using the presented algorithm starting at position $n=30$, as starting at $n=1$ led to numerical problems. As $\tilde\lambda>\lambda$ the feedback is positive. Note that the distribution of proteins in the unbound state shows two maxima. The bimodality of the distribution of the sum of proteins in both states was already demonstrated and discussed elsewhere \citep{hornos2005self,ramos2011exact}, albeit in this example it is shown solely for the unbound state. The unbound state is relatively short, so its protein distribution depends strongly on the protein distribution of the bound state. The left maximum is caused mainly by proteins at the end of the unbound state, whereas the right maximum is caused by proteins at the beginning of the unbound state.\newline
\begin{figure}[htbp]
\centering
\includegraphics[scale=0.8]{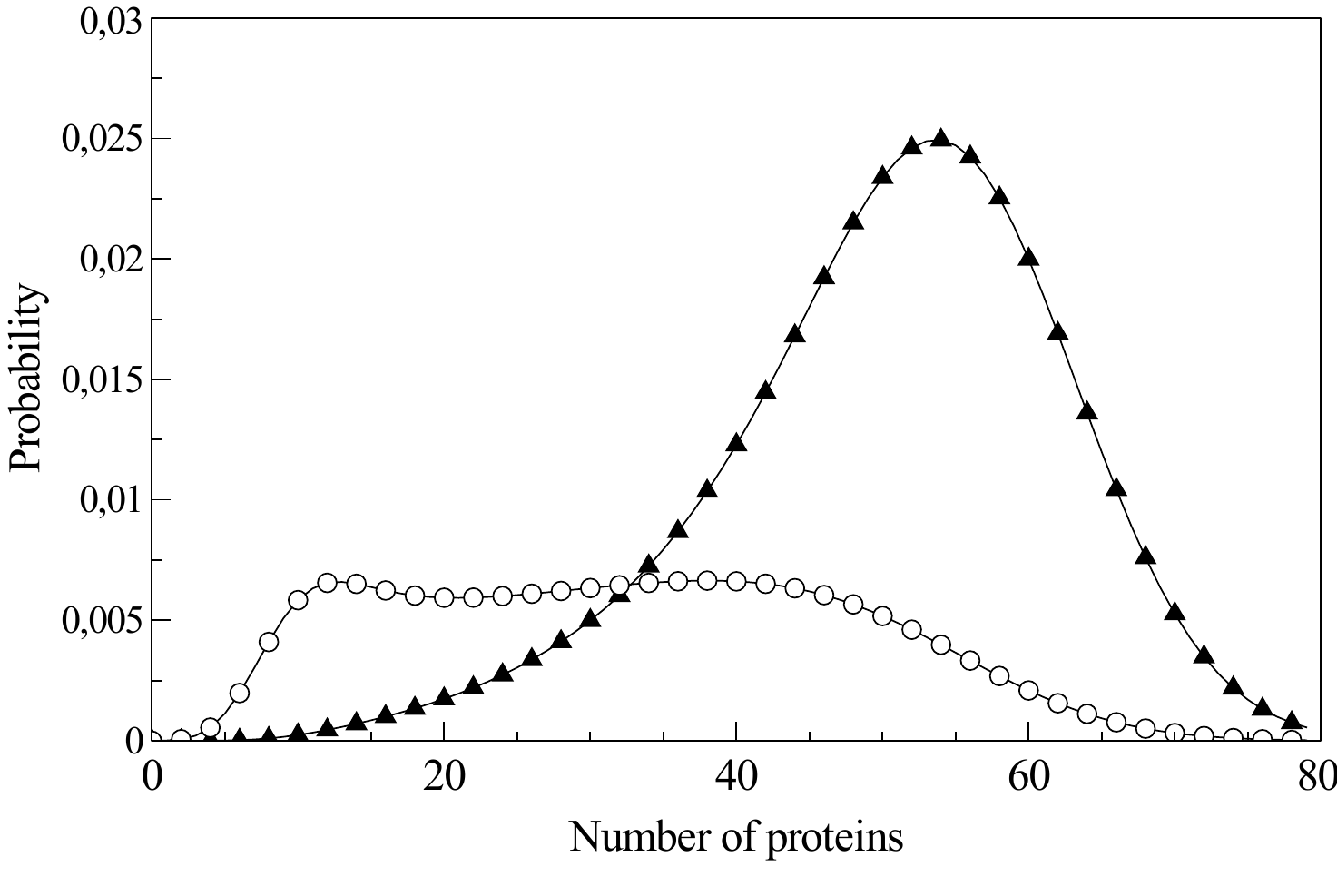}
\caption{Distribution of $\left(N,A\right)$ (circles for $A=1$ and triangles for $A=0$). Parameter values are $\lambda=6,\tilde\lambda=60,\mu=1,\theta=1, \psi=0.066$. Only even protein numbers are displayed.}%
\label{pic:Fig2}%
\end{figure}
Different extensions of the model are possible. So the gene can have more than two expression states and the produced protein may have to be further metabolized for example by dimerization to bind to the promoter.
As mass-action kinetics are assumed for each reaction in the model, not all important features of gene expression may be captured.\newline
In Section \ref{feedback_impact} the influence of feedback on the system is considered. Given that a gene is regulated by its product, and this product is only produced by this gene, the features of feedback compared to an equivalent situation, in which proteins are produced at the same rate but the state of the gene is switched independent of the protein number, is considered.\newline
\begin{figure}[htbp]
\centering
\includegraphics[scale=0.75]{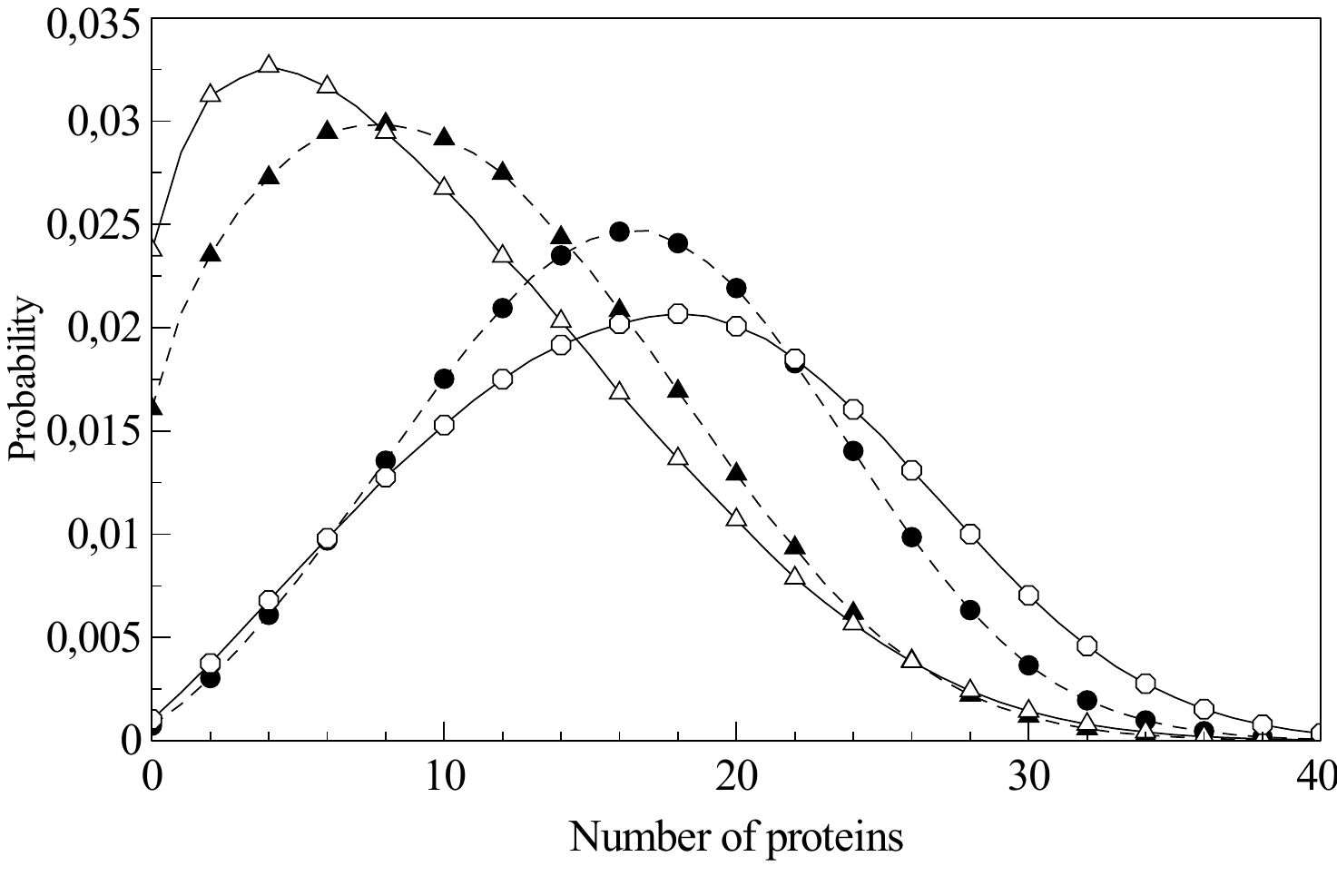}
\caption{Steady state distribution of the MFM (full markers, dashed line) and the NFM (hollow markers). Circles mark the unbound state and triangles mark the bound state. Only even protein numbers are displayed. Parameter values are $\lambda=30,\tilde\lambda=0,\mu=1,\theta=1.3$ and $\psi=0.1$. $\tilde\psi$ is choosen as described in \eqref{psitilde}, hence $P\left\{\tilde{A}=1\right\}=P\left\{\bar{A}=1\right\}$.}
\label{pic:Fig4}
\end{figure}
In figure \ref{pic:Fig4} both situations are compared. It is visible, that in the NFM the total protein variance and the expected number proteins in the unbound state are higher.\newline 
For the MFM and the NFM algorithms similar to the one described in Section \ref{section_algo} can be used to determine the distributions of $\bar{N}$ and $M$ repectively.\newline
Given the same total protein production rate, and assuming that besides gene expression there is no influx of protein to the cytoplasma, the models suggest that non-feedback leads to more (less) distinguishable distributions of protein numbers in the two states, if $\lambda>\tilde\lambda$ ($\lambda<\tilde\lambda$). So, if the protein concentration should be low in the unbound and high in the bound state, non-feedback is superior in this framework.

\subsection{Method of Hornos}
Previous work \citep{hornos2005self,ramos2011exact} give an exact
solution for the protein distribution in the steady state in this model framework.
This is done by rewriting the master equation as partial differential equation, which
can be totally integrated using the confluent Heun function.
The probability $P\left\{X=(n,j)\right\}$, with $n\geq 2$ and $j\in\left\{0,1\right\}$, can be computed. For
\begin{align*}
&a:=1+\frac{\theta}{\mu+\psi}\left(1+\frac{\psi\lambda}{\mu\lambda-(\mu+\psi)\tilde\lambda}\right),\\
&b:=1+\frac{\theta}{\mu+\psi}+\frac{\psi\lambda}{\left(\mu+\psi\right)^2},\\
&\eta(z):=-\frac{\left[\tilde\lambda\left(1+\frac{\psi}{\mu}\right)-\lambda\right]\left[(\mu+\psi)z-\mu\right]}{\left(\mu+\psi\right)^2}
\end{align*}
it holds $$P\left\{X=(n,1)\right\}=\frac{A}{n!}\sum_{s=0}^{n}\binom{n}{s}\left(\frac{\tilde\lambda}{\mu}\right)^{n-s}\left(\frac{d\eta}{dz}\right)^{s}\frac{(a)_s}{(b)_s}M\left(a+s,b+s,\eta(0)\right),$$ where $(.)_{n}$ and $M$ denote the rising factorial and the Kummer function \citep{abramowitz1970handbook}. $A$ is a normalization constant, guaranteeing, that the sum of the probabilities is $1$.\newline
So one has to evaluate $n$ Kummer functions. However evaluating the Kummer function is numerical sophisticated \citep{muller2001computing}. 
In contrast the recursion method is easier to apply and more intuitive.\newline
Comparing the recursion method with the method of Hornos in numerous simulations, it can be concluded, that both have their advantages and disadvantages. Partly they can complement each other. So the method of Hornos works well for small $\frac{\lambda+\tilde\lambda}{2\mu}$ but a wide range of $\theta$ and $\psi$, due to the numerical  evaluation of the Kummer function, whereas the recursion method works well for a wide range of $\frac{\lambda+\tilde\lambda}{2\mu}$, but only for a relatively small range of $\theta$ and $\psi$. The recursion method has the disadvantage, that minimal changes of probabilities, which are close to machine precision, might lead to enormous changes in the outcome of the recursion. Figure \ref{Algo_Vergleich_feedback} demonstrates the areas, in which both algorithms work well for two examples. In this context an algorithm works well, if the resulting distribution is non negative and asymptotic to zero for large numbers of proteins. The distribution of protein-numbers is evaluated between $0$ and the triple of the smallest integer, below which the majorant distribution, which is Poisson, has at least $99.9\%$ of its mass.
Hence this integer is the smallest $n\in\mathbb{N}$, for which holds
\begin{equation*}
\begin{split}
\mathrm{e}^{-\frac{\max\left(\lambda,1+\tilde\lambda\right)}{\mu}} \sum_{k=0}^n \frac{\max\left(\lambda,1+\tilde\lambda\right)^k}{\mu^kk!}\geq 0.999.
\end{split}
\end{equation*}
\begin{figure}
\subfigure[Parameters are $\mu=1$, $\theta=0.5$ and $\tilde\lambda=1.5\lambda$]{\includegraphics[width=0.49\textwidth]{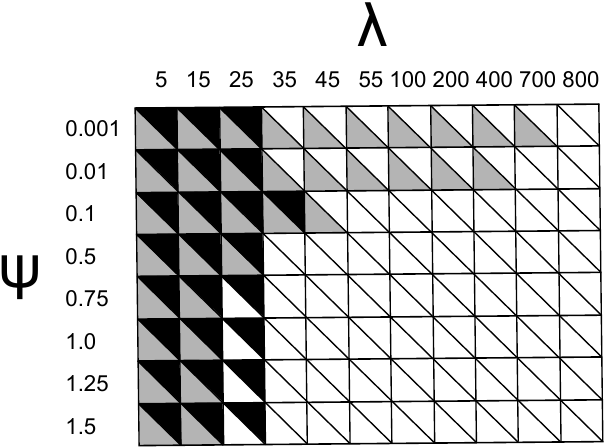}}\quad
\subfigure[Parameters are $\mu=1$, $\theta=0.1$ and $\tilde\lambda=0.5\lambda$]{\includegraphics[width=0.49\textwidth]{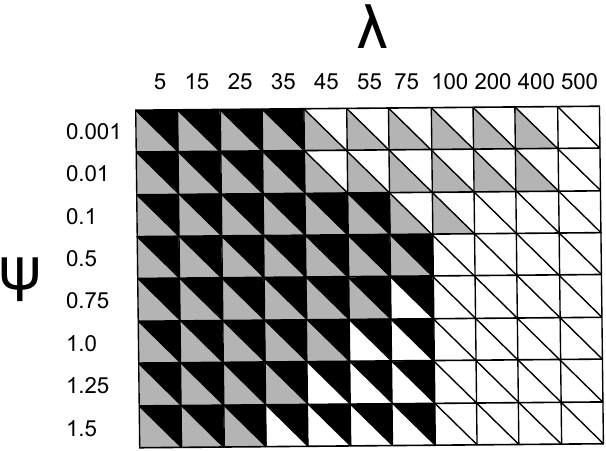}}\caption[Comparism of both algorithms]{Comparism of both algorithms. An upper black triangle indicates, that the algorithm proposed by Hornos works well, a lower grey triangle indicates, that the algorithm described in Section \ref{section_algo} works well}\label{Algo_Vergleich_feedback}
\end{figure}

\newpage
\section{Proofs}\label{section_proof}
\subsection{Proof Theorem \ref{T1}}
The equation \eqref{E1} can be derived by the master equation in the steady state inductively.\newline
The equation \eqref{E2} can be derived by summation over all $n$ of the master equation in the steady state for $A=1$ or $A=0$.
The equations \eqref{E3} and \eqref{E4} are direct consequences of the master equation evaluated in the steady state.

\subsection{Proof Corollary \ref{C1}}
Summation over $n$ of \eqref{E1} yields to the equation \eqref{E5}. The equation \eqref{E6} follows directly from \eqref{E2}. The equation \eqref{E7} follows directly from \eqref{E5} and \eqref{E6}. The recursions \eqref{E7}, \eqref{E8} and \eqref{E9} are direct consequences from \eqref{E1}, \eqref{E3} and \eqref{E4}.
\eqref{E02} follow directly from \eqref{E5}.

\subsection{Proof Theorem \ref{T2}}
For the proof the following lemma is needed:
\begin{lemma}
	\label{L1}
	Let $s\in\mathbb{N}$ and $Y$ be a positive discrete random variable with $\mathbb{E}\left[Y^s\right]<\infty$, it holds that
\begin{equation*}
\begin{split}
\sum_{i=1}^{\infty}i^{s}(i+1)P\left\{Y=i+1\right\}=\sum_{j=0}^{s}\binom{s}{j}(-1)^{j}\mathbb{E}\left[Y^{s-j+1}\right].
\end{split}
\end{equation*}
\end{lemma}
\begin{proof}
If $\mathbb{E}\left[Y^s\right]<\infty$, it follows that
\begin{equation*}
\begin{split}
&\sum_{i=1}^{\infty}i^{s}(i+1)P\left\{Y=i+1\right\}=\sum_{i=2}^{\infty}(i-1)^{s}iP\left\{Y=i\right\}\\
=&\sum_{i=2}^{\infty}\sum_{j=0}^{s}\binom{s}{j}i^{s-j+1}(-1)^{j}P\left\{Y=i\right\}\\
=&\sum_{j=0}^{s}\binom{s}{j}(-1)^{j}\left[\sum_{i=0}^{\infty}i^{s-j+1}P\left\{Y=i\right\}-P\left\{Y=1\right\}\right]\\
=&\sum_{j=0}^{s}\binom{s}{j}(-1)^{j}\left[\mathbb{E}\left[Y^{s-j+1}\right]-P\left\{Y=1\right\}\right]=\sum_{j=0}^{s}\binom{s}{j}(-1)^{j}\mathbb{E}\left[Y^{s-j+1}\right].\\
\end{split}
\end{equation*}
\end{proof}
For the sake of brevity set $\pi_{(n,i)}:=P\left\{X=(n,i)\right\}$ in this proof for $n\in\mathbb{N}, i\in\left\{0,1\right\}$. It holds that
\begin{equation*}
\begin{split}
& \mathbb{E}\left[1_{A=1}N^s\right]=\sum_{i=0}^{\infty}i^s\pi_{(i,1)}\\
 \overset{\text{\eqref{E1}}}{=}& \frac{1}{\lambda}\left[\sum_{i=1}^{\infty} i^{s}\left(i\mu\pi_{(i+1,1)}+i\mu\pi_{(i+1,0)}+\mu\pi_{(i+1,1)}-\tilde{\lambda}\pi_{(i,0)}\right)\right]\\
=&\frac{\mu}{\lambda}\sum_{i=1}^{\infty}(i-1)^s\left((i-1)\left(\pi_{(i,1)}+\pi_{(i,0)}\right)+\pi_{(i,1)}\right)-\frac{\tilde{\lambda}}{\lambda}\mathbb{E}\left[N^{s}1_{A=0}\right]\\
=&\frac{\mu}{\lambda}\left(\sum_{j=0}^{s}\binom{s}{j}(-1)^j\left(\mathbb{E}\left[N^{s-j+1}\right]-\mathbb{E}\left[N^{s-j}1_{A=0}\right]\right)\right)-\frac{\tilde{\lambda}}{\lambda}\mathbb{E}\left[N^{s}1_{A=0}\right].\\
\end{split}
\end{equation*}
It follows for the $(s+1)$-th moment, that
\begin{equation*}
\begin{split}
\mathbb{E}\left[N^{s+1}\right]=&\frac{\lambda}{\mu}\mathbb{E}\left[N^{s}1_{A=1}\right]+\mathbb{E}\left[N^{s}1_{A=0}\right]+\frac{\tilde{\lambda}}{\mu}\mathbb{E}\left[N^{s}1_{A=0}\right]\\
&-\sum_{j=1}^{s}\binom{s}{j}(-1)^j\left(\mathbb{E}\left[N^{s-j+1}\right]-\mathbb{E}\left[N^{s-j}1_{A=0}\right]\right).\\
\end{split}
\end{equation*}
Next $\mathbb{E}\left[1_{A=1}N^s\right]$ is derived. It holds
\begin{equation*}
\begin{split}
& \mathbb{E}\left[1_{A=0}N^s\right]=\sum_{i=0}^{\infty}i^s\pi_{(i,0)}=\sum_{i=1}^{\infty}i^s\pi_{(i,0)} \\
&\overset{\text{\eqref{E3}}}{=} \frac{1}{\theta}\left[\sum_{i=1}^{\infty}i^{s+1}(\mu+\psi)\pi_{(i,1)}+i^s\pi_{(i,1)}\lambda-\mu i^s(i+1)\pi_{(i+1,1)}-\lambda i^s\pi_{(i-1,1)}\right] \\
&= \frac{1}{\theta}\Big[\mathbb{E}\left[1_{A=1}N^{s+1}\right](\mu+\psi)+\lambda\mathbb{E}\left[1_{A=1}N^s\right]\\
& \qquad \qquad \qquad \qquad -\sum_{i=1}^{\infty}\left[\mu i^s(i+1)\pi_{(i+1,1)}+\lambda i^s\pi_{(i-1,1)}\right]\Big] \\
& \overset{\text{(Lemma \ref{L1})}}{=}\frac{1}{\theta}\Big[\mathbb{E}\left[1_{A=1}N^{s+1}\right]\psi\\
& \qquad \qquad \qquad \qquad -\sum_{j=1}^{s}\binom{s}{j}\left((-1)^{j}\mu \mathbb{E}\left[1_{A=1}N^{s-j+1}\right]+\lambda\mathbb{E}\left[1_{A=1}N^{s-j}\right]\right)\Big]. \\    
\end{split}
\end{equation*}
Hence it follows
\begin{equation*}
\begin{split}
&\mathbb{E}\left[1_{A=1}N^{s+1}\right]\\
&=\frac{1}{\psi}\left[\theta \mathbb{E}\left[1_{A=0}N^s\right]+\sum_{j=1}^{s}\binom{s}{j}\left((-1)^{j}\mu \mathbb{E}\left[1_{A=1}N^{s-j+1}\right]+\lambda\mathbb{E}\left[1_{A=1}N^{s-j}\right]\right)\right].\\    
\end{split}
\end{equation*}

\subsection{Proof Theorem \ref{T3}}
It is straighforward to prove, that the steady state distribution of the Markov chain exists and that it is unique. The recursion is a direct consequence of the master equation in the steady state and the master equation in the steady state can be derived by the recursions. Hence there is a one-to-one correspondence between the recursion and the corresponding Markov chain. So due to the linearity of the recursion, if a sequence can be derived using the recursions, for which the corresponding sequence of partial sums is bounded, it is proportional to the unique steady state solution.

\subsection{Proof Corollary \ref{R0}}
As the bound states expected duration is $\theta$ and the unbound states expected duration is $\mathbb{E}\left[N|A=1\right]\psi$ it holds 
\begin{equation*}
P\left\{A=1\right\}=\frac{\theta}{\theta+\mathbb{E}\left[N|A=1\right]\psi}.
\end{equation*}
If $\tilde\lambda$ equals $\lambda$, the distribution of proteins during the bound state, and the distribution of proteins minus one during the bond state are Poisson with parameter $\frac{\lambda}{\mu}$. It holds that, if the distribution of a random variable $Y$ satisfy the condition 
\begin{equation}
\begin{split}\label{Poisson_feature}
P\left\{Y=n\right\}=\frac{\lambda}{\mu n}P\left\{Y=n-1\right\},
\end{split}
\end{equation}
$Y$ is Poisson distributed with parameter $\frac{\lambda}{\mu}$.
It can be shown that the recursions \eqref{E8} and \eqref{E9} satisfy for $\lambda=\tilde\lambda$ the conditions $P\left\{X=(n,1)\right\}=\frac{\lambda}{\mu n}P\left\{X=(n-1,1)\right\}$ and $P\left\{X=(n,0)\right\}=\frac{\lambda}{\mu (n-1)}P\left\{X=(n-1,0)\right\}$. Hence they are Poisson as described above and by Theorem \ref{T3} they are unique solutions.
As consequence it follows, that
\begin{equation*}
\mathbb{E}\left[N|A=1\right]=\mathbb{E}\left[N|A=0\right]-1=\frac{\lambda}{\mu}\mbox{, hence } P\left\{A=1\right\}=\frac{\mu\theta}{\mu\theta+\lambda\psi}.
\end{equation*}
To prove \eqref{ineq}, all parameters but $\tilde\lambda$, which is increased starting from $\tilde\lambda=\lambda$, are fixed.
It holds for the unbound and the bound state, that the higher the protein, with which the state starts, the higher the expected protein number at the end of the state. For the bound state this is easy to verify, with Theorems \ref{T4} and \ref{T5}. In the unbound state the number of proteins during the state is given as the sum of the two independent processes. There are proteins, which existed at the beginning of the state and decay during the state, and proteins, which are produced in the state and yet not decayed. Hence, if the impact of the number of proteins $n\in\mathbb{N}$ at time $0$ of the state on the expected number of proteins at the end of the state is analysized, only the first process, which is binomial distributed after time $t$ with rate $\exp(-\mu t)$ and $n$ trials, has to be considered. It is straighforward to verify, that the higher $n$, the shorter the expected time in the unbound state. Furthermore the distribution of the number of proteins at the end of both state, which started with $m\le n$ proteins is dominated by the corresponding distributions, which started with $n$ proteins. Given that it can be verified, that, if $\tilde\lambda$ is increased $\mathbb{E}\left[N|A=0\right]$ increases.
Using \eqref{E6} and \eqref{eq_lambda} the inequality $P\left\{A=1\right\}\leq\frac{\mu\theta}{\mu\theta+\max(\lambda,\tilde\lambda)\psi}$ follows. The lower estimate is proven analogously.

\subsection{Proof Theorem \ref{T4}}
If the switching of states is neglected, in both states $A=1$ and $A=0$ there are two independent processes regulating the protein-number. At the one hand each protein decays after an exponential distributed waiting time with rate $\mu$. This affects all proteins in the state $A=1$ and all but one protein in the state $A=0$. At the other hand there is an exponential distributed production with rate $\lambda$ and $\tilde{\lambda}$ respectively. Let $n$ be the number of proteins at the beginning of a state.
The distribution of the number of non-decayed proteins at time $t$, given $n$ proteins at time $0$ is binomial with rate $\exp(-\mu t)$ and $n$ trials in the state $A=1$. For the state $A=0$ there are $n-1$ trials and the gene-bound protein has to be added to the protein number.\newline
The distribution of the number of produced and not yet decayed proteins at time $t$ is Poisson with rate $\frac{X}{\mu}\left(1-\exp(-\mu t)\right)$ with $X=\lambda$ in state $A=1$ and $X=\tilde{\lambda}$ in state $A=0$.
The bound state duration is distributed exponentially with parameter $\theta$. Hence (starting at $t=0$) $\exp(-t\theta)\theta$ is the density of the duration of the bound state. Let $Q_{n,t}$ be a random variable describing the number of proteins at time $t$ in the bound state, given that $n$ proteins existed at the beginning of the bound state at time $0$ and that the bound state is yet not stopped.
It follows, that
\begin{equation*}
\begin{split}
&\mathbb{E}\left[Q_{n,t}\right]=(n-1)e^{-t\mu}+1+\frac{\tilde{\lambda}}{\mu}\left(1-e^{-t\mu}\right).\\
\end{split}
\end{equation*}
Hence \eqref{B_eq1} holds, as
\begin{equation*}
\begin{split}
&\mathbb{E}\left[{C_1|B_0=n}\right]=\mathbb{E}\left[\int_{0}^{\infty}Q_{n,t}e^{-t\theta}\theta dt\right]\\
&=\int_{0}^{\infty}\mathbb{E}\left[Q_{n,t}\right]e^{-t\theta}\theta dt=\left(n-\frac{\tilde{\lambda}}{\mu}-1\right)\frac{\theta}{\theta+\mu}+1+\frac{\tilde{\lambda}}{\mu}=\frac{n\theta+\tilde\lambda+\mu}{\theta+\mu}.
\end{split}
\end{equation*}
The remaining propositions are direct consequences of these features and the features of the binomial- and Poisson-distribution.

\subsection{Proof Theorem \ref{T5}}
As the duration of the bound state is exponential distributed, it holds that at each time point in the bound state the propensity of switching to the unbound state is the same. Hence the expected protein number during a bound state equals the expected protein number at the end of a bound state. 

\subsection{Proof Theorem \ref{T6}}
The first equation follows directly from \eqref{B_eq2}.\newline
Using the notation from the proof of Theorem \ref{T4}, it holds
\begin{equation*}
\begin{split}
\mathbb{E}\left[C^2\right]=\mathbb{E}\left[\frac{\int_{0}^{\infty}\exp(-\theta t)\theta\int_{0}^{t}Q_{n,s}^2dsdt}{\int_{0}^{\infty}t\exp(-t\theta)\theta dt}\right]=\theta^2\int_{0}^{\infty}\exp(-t\theta)\int_{0}^{t}\mathbb{E}\left[Q_{n,s}^2\right]dsdt.
\end{split}
\end{equation*} 
By the proof of Theorem $\ref{T4}$ it holds for $n\geq1$ $$Q_{n,t}=X_{n-1}(t)+1+Y(t).$$
$X_n(t)$ is binomial distributed with rate $\exp(-\mu t)$ and $n$ trials and $Y(t)$ is Poisson distributed with rate $\frac{\tilde\lambda}{\mu}\left(1-\exp(-\mu t)\right)$.
Using the features of the binomial and Poisson distribution it follows for $n\in\mathbb{N}$
\begin{equation*}
\begin{split}
&\mathbb{E}\left[Q_{n,t}^2\right]\\
&=\mathbb{E}\left[(1+X_{n-1}(t))^2\right]+2\mathbb{E}\left[(1+X_{n-1}(t))Y(t)\right]+\mathbb{E}\left[Y(t)^2\right]\\
&=1+2\mathbb{E}\left[X_{n-1}(t)\right]+\mathbb{E}\left[X_{n-1}(t)^2\right]+2\mathbb{E}\left[Y(t)\right]+2\mathbb{E}\left[X_{n-1}(t)Y(t)\right]+\mathbb{E}\left[Y(t)^2\right]\\
&=1+2(n-1)e^{-\mu t}+(n-1)e^{-\mu t}(1-e^{-\mu t})+\left((n-1)e^{-\mu t}\right)^2+2\frac{\tilde{\lambda}}{\mu}\left(1-e^{-\mu t}\right)\\
&\qquad \qquad+\frac{\tilde{\lambda}}{\mu}\left(1-e^{-\mu t}\right)+\left(\frac{\tilde{\lambda}}{\mu}\left(1-e^{-\mu t}\right)\right)^2+2\frac{\tilde{\lambda}}{\mu}\left(1-e^{-\mu t}\right)(n-1)e^{-\mu t}\\
&=1+3(n-1)e^{-\mu t}-(n-1)e^{-2\mu t}+\left((n-1)e^{-\mu t}\right)^2+3\frac{\tilde{\lambda}}{\mu}\left(1-e^{-\mu t}\right)\\
&\qquad \qquad+\left(\frac{\tilde{\lambda}}{\mu}\left(1-e^{-\mu t}\right)\right)^2+2\frac{\tilde{\lambda}}{\mu}\left(1-e^{-\mu t}\right)(n-1)e^{-\mu t}\\
&=n^2e^{-2t\mu}+n\left(3+\frac{2\tilde{\lambda}}{\mu}\right)\left(e^{-t\mu}-e^{-2t\mu}\right)+1+\frac{3\tilde{\lambda}}{\mu}+\left(\frac{\tilde{\lambda}}{\mu}\right)^2\\
&\qquad \qquad-e^{-\mu t}\left(3+\frac{5\tilde{\lambda}}{\mu}+2\left(\frac{\tilde{\lambda}}{\mu}\right)^2\right)+e^{-2\mu t}\left(2+\frac{2\tilde{\lambda}}{\mu}+\left(\frac{\tilde{\lambda}}{\mu}\right)^2\right).\\
\end{split}
\end{equation*}
Given that it can be computed
\begin{equation*}
\begin{split}
&\int_{0}^{\infty}e^{-\theta t}\int_{0}^{t}\mathbb{E}\left[Q_{n,s}^2\right]dsdt=\\
&\int_{0}^{\infty}e^{-\theta t}\int_{0}^{t}\Bigg(n^2e^{-2s\mu}+n\left(3+\frac{2\tilde{\lambda}}{\mu}\right)\left(e^{-s\mu}-e^{-2s\mu}\right)+1+\frac{3\tilde{\lambda}}{\mu}+\left(\frac{\tilde{\lambda}}{\mu}\right)^2\\
&\qquad -e^{-\mu s}\left(3+\frac{5\tilde{\lambda}}{\mu}+2\left(\frac{\tilde{\lambda}}{\mu}\right)^2\right)+e^{-2\mu s}\left(2+\frac{2\tilde{\lambda}}{\mu}+\left(\frac{\tilde{\lambda}}{\mu}\right)^2\right)\Bigg)dsdt\\
&=n^2\frac{1}{\theta\left(\theta+2\mu\right)}+n\left(3+\frac{2\tilde{\lambda}}{\mu}\right)\left(\frac{1}{\theta\left(\theta+\mu\right)}-\frac{1}{\theta\left(\theta+2\mu\right)}\right)+1+\frac{3\tilde{\lambda}}{\mu}+\left(\frac{\tilde{\lambda}}{\mu}\right)^2\\
&-\frac{1}{\theta\left(\theta+\mu\right)}\left(3+\frac{5\tilde{\lambda}}{\mu}+2\left(\frac{\tilde{\lambda}}{\mu}\right)^2\right)+\frac{1}{\theta\left(\theta+2\mu\right)}\left(2+\frac{2\tilde{\lambda}}{\mu}+\left(\frac{\tilde{\lambda}}{\mu}\right)^2\right).\\
\end{split}
\end{equation*}
Hence it holds that
\begin{equation*}
\begin{split}
&\mathbb{E}\left[C^2\right]=
\mathbb{E}\left[B\right]\left(3+\frac{2\tilde{\lambda}}{\mu}\right)\left(\frac{\theta\mu}{\left(\theta+\mu\right)\left(\theta+2\mu\right)}\right)+\theta^2\left(1+\frac{3\tilde{\lambda}}{\mu}+\left(\frac{\tilde{\lambda}}{\mu}\right)^2\right)\\
&-\frac{\theta}{\left(\theta+\mu\right)}\left(3+\frac{5\tilde{\lambda}}{\mu}+2\left(\frac{\tilde{\lambda}}{\mu}\right)^2\right)+\frac{\theta}{\left(\theta+2\mu\right)}\left(2+\frac{2\tilde{\lambda}}{\mu}+\left(\frac{\tilde{\lambda}}{\mu}\right)^2\right)+\frac{\theta\mathbb{E}\left[B^2\right]}{\left(\theta+2\mu\right)},\\
&\mbox{so as consequence it holds }\\
&\mathbb{E}\left[B^2\right]=\frac{\theta+2\mu}{\theta+\mu}\left(3+5\frac{\tilde{\lambda}}{\mu}+2\left(\frac{\tilde{\lambda}}{\mu}\right)^2\right)-\left(2+2\frac{\tilde{\lambda}}{\mu}+\left(\frac{\tilde{\lambda}}{\mu}\right)^2\right)\\
&+\mathbb{E}\left[C^2\right]\frac{\theta+2\mu}{\theta}-\frac{\mu\mathbb{E}\left[B\right]}{\theta+\mu}\left(3\mu \theta+2\frac{\tilde{\lambda}}{\mu}\right)-(\theta+2\mu)\theta\left(1+3\frac{\tilde{\lambda}}{\mu}+\left(\frac{\tilde{\lambda}}{\mu}\right)^2\right).\\
\end{split}
\end{equation*}

\subsection{Proof Theorem \ref{T7}}
The probability, that starting with $n$ proteins all proteins decay in an unbound state, before reaching the bound state is $\left(\frac{\mu}{\psi+\mu}\right)^{n}$, conversely for $1\leq m\leq n$ the conditional probability for reaching the bound state with $m$ proteins is $P\left\{B_1=m|C_0=n\right\}=\frac{\psi}{\psi+\mu}\left(\frac{\mu}{\psi+\mu}\right)^{n-m}$.\newline
For $1\leq m\leq n$ the expected time for the event that $n$ proteins decay to $m$ proteins in the unbound state before it switches to the bound state is $\frac{1}{\mu}\sum_{i=m+1}^{n}\frac{1}{i}+\frac{1}{m\psi}$. The expected time in the bound state is $\frac{1}{\theta}$. With this the recursions for $P\left\{S=i|X_0=(n,0)\right\}$ and $\mathbb{E}\left[T_i|X_0=(n,0),i<S\right]$ and the expectation values $\mathbb{E}\left[S|X_0=(n,0)\right]$ and $\mathbb{E}\left[S^2|X_0=(n,0)\right]$ can be derived.\newline
To prove the last claim, features of the geometric series are used. It holds for $n\geq 1$, that
\begin{equation*}
\begin{split}
\mathbb{E}\left[B_1|C_0=n\right]&=\sum_{i=0}^{n}i\frac{\psi}{\psi+\mu}\left(\frac{\mu}{\mu+\psi}\right)^{i}\\
&=\frac{(\mu+\psi)^2}{\psi\mu}\left[n\left(\frac{\mu}{\mu+\psi}\right)^{n+2}-(n+1)\left(\frac{\mu}{\mu+\psi}\right)^{n+1}+\frac{\mu}{\mu+\psi}\right]\\
&=\frac{\mu+\psi}{\psi}\left(1-\left(\frac{\mu}{\mu+\psi}\right)^{n}\left(1+\frac{n\psi}{\mu+\psi}\right)\right).
\end{split}
\end{equation*}

\subsection{Proof Corollary \ref{C4}}
$\mathbb{E}\left[C_2|C_1=n\right]$ can be computed using the expression for $\mathbb{E}\left[B_1|C_0=n\right]$ from Theorem \ref{T7}. Let $B_0$ be the distribution of proteins at the end of the unbound state, between the bound states ending with $C_1$ and $C_2$. It holds
\begin{equation*}
\begin{split}
\mathbb{E}\left[C_2|C_1=n\right]&=\sum_{m=1}^{\infty}mP\left\{C_2=m|C_1=n\right\}\\
&=\sum_{b=1}^{n}P\left\{B_0=b|C_1=n\right\}\sum_{m=1}^{\infty}mP\left\{C_2=m|B_0=b\right\}\\
&\overset{\text{\eqref{B_eq1}}}{=}\sum_{b=1}^{n}P\left\{B_0=b|C_1=n\right\}\left(\frac{b\theta+\tilde\lambda+\mu}{\theta+\mu}\right)\\
&=\frac{\theta}{\theta+\mu}\mathbb{E}\left[B_0|C_1=n\right]+\frac{\tilde\lambda+\mu}{\theta+\mu}.
\end{split}
\end{equation*}
Monotony of $\mathbb{E}\left[C_2|C_1=n\right]$ is straightforward to show. So for $n\geq 1$ it holds
$$\mathbb{E}\left[C_2|C_1=1\right] \leq \mathbb{E}\left[C_2|C_1=n\right] \leq \lim_{m\rightarrow\infty}\mathbb{E}\left[C_2|C_1=m\right],$$
$$\lim_{m\rightarrow\infty}\mathbb{E}\left[C_2|C_1=m\right]=\frac{\tilde\lambda+\mu}{\theta+\mu}+\frac{\theta}{\theta+\mu}\frac{\mu+\psi}{\psi},$$
$$\mathbb{E}\left[C_2|C_1=1\right]=\frac{\tilde\lambda+\mu}{\theta+\mu}+\frac{\theta}{\theta+\mu}\frac{\psi}{\mu+\psi}.$$
The distribution of the totally produced proteins is easy to compute given the number of bound states. During one bound state $m$ proteins are produced with probability $\frac{\theta}{\theta+\tilde\lambda}\left(\frac{\tilde\lambda}{\theta+\tilde\lambda}\right)^m$. $\binom{m+i-1}{i}$ is the number of multisets of cardinality $i$, hence the number of possibilities to express $m\in\mathbb{N}$ as sum of $i$ positive integers including zero. For the distribution after $i$ states, one has to take $\frac{\theta}{\theta+\tilde\lambda}\left(\frac{\tilde\lambda}{\theta+\tilde\lambda}\right)^m$ $\binom{m+i-1}{i}$ times.

\subsection{Proof Proposition \ref{Comparism_Prop}}
Let $z:=\frac{\tilde\psi}{\psi}$, using \eqref{ENA1} and \eqref{psitilde} it can be derived, that $\mathbb{E}\left[\bar{N}|A=1\right]=z$. Consider $C$, the number of proteins at the beginning of the unbound state. If $C$ is transformed by an unbound state of the NFM, it holds that for all protein numbers smaller than $z$ the hazard of ending the unbound state is greater than in the MFM.
As the expected duration of the unbound state in both models is assumed equal (and does not depend on the start position in the NFM) these higher hazard must be compensated by protein numbers greater $z$, which have a lower hazard of ending the unbound state in the NFM. Hence the expected number of proteins after the unbound state in the NFM is greater than $z$. Furthermore by \eqref{EC}, the expected number of proteins after (and during) an unbound state in the NFM is linear dependent to the expected number of proteins at the beginning of this unbound state. Hence, if $\mathbb{E}\left[M|\tilde{A}=0\right]\geq\mathbb{E}\left[\bar{N}|\bar{A}=0\right]$, it follows $\mathbb{E}\left[M|\tilde{A}=1\right]>\mathbb{E}\left[\bar{N}|\bar{A}=1\right]$, which contradicts $\mathbb{E}\left[\bar{N}\right]=\mathbb{E}\left[M\right]$. Hence the equation \eqref{Comparism_exp} holds.\newline
%Let $B$ be the the number of proteins at the end of the unbound state in the steady state in the MFM. As the relation \eqref{EC} holds for both models, it follows using Theorem \ref{T5}, that $\mathbb{E}\left[M|\tilde{A}=1\right]\leq\mathbb{E}\left[B\right]$.\newline
\eqref{Comparism_var} follows directly by assuming \eqref{psitilde} and using \eqref{M_2nd_moment}, \eqref{N_2nd_moment} and the fact that there is a $\epsilon\geq0$, so that
\begin{equation*}
\begin{split}
&\mathbb{E}\left[1_{\tilde{A}=1}M\right]=\mathbb{E}\left[1_{\bar{A}=1}\bar{N}\right]+\epsilon,\\
&\mathbb{E}\left[1_{\tilde{A}=0}M\right]=\mathbb{E}\left[1_{\bar{A}=0}\bar{N}\right]-\epsilon
\end{split}
\end{equation*}
holds. 

\subsection{Proof Proposition \ref{impact_Prop}}

Set $c:=\frac{\lambda}{\mu}$ constant. The effect of changing $\mu$ (and consequently $\lambda=c\mu$) on the distribution $X=(N,A)$ is considered. If $\mu$ tends to $\infty$, $\mathbb{E}\left[N|A=1\right]$ tends to $c$. Thus it holds
\begin{equation}
\begin{split}\label{PA_1_Grenzwert_0}
\lim_{\mu\rightarrow\infty}P\left\{A=1\right\}=\frac{\theta}{\theta+c\psi}
\end{split}
\end{equation}
This is by \eqref{ineq} the lowest possible value, which $P\left\{A=1\right\}$ can attain.
Using \eqref{E3} it can be seen, that for all $n\in\mathbb{N}$ 
$$ \lim_{\mu\rightarrow0}P\left\{X=(n,0)\right\}=\frac{\psi n}{\theta}P\left\{X=(n,1)\right\} $$
Hence as $P\left\{X=(0,0)\right\}=0$, it always holds $\mathbb{E}\left[N|A=1\right]<\mathbb{E}\left[N\right]<\mathbb{E}\left[N|A=0\right]$.

Using \eqref{E6} it follows
\begin{equation}
\begin{split}\label{PA_1_Grenzwert}
\lim_{\mu\rightarrow0}P\left\{A=1\right\}> \frac{\theta}{\theta+\mathbb{E}\left[N\right]\psi}\geq \frac{\theta}{\theta+c\psi}
\end{split}
\end{equation}
Let $\tilde{\lambda}=0$, set $$E_0:=\frac{1}{2\psi}\left(\psi-\theta+\sqrt{(\psi-\theta)^2+4\theta\psi c}\right)=\mathbb{E}\left[\tilde{N}\right].$$
Hence $E_0$ is constant, if the ratio $c$ and $\theta,\psi$ are fixed. If $\mu$ tends to zero, it can be derived by using \eqref{E5} and \eqref{PA_1_Grenzwert}, that 
$$\lim_{\mu\rightarrow0}\mathbb{E}\left[N\right]\geq E_0.$$
Furthermore using \eqref{PA_1_Grenzwert_0} and \eqref{E5} it holds
$$\lim_{\mu\rightarrow\infty}\mathbb{E}\left[N\right]= E_{\infty}:=c\frac{\theta+\psi}{\theta+c\psi}$$
Hence $E_{0}- E_{\infty}$ is greater equal the maximal distance between both approaches.

\renewcommand{\abstractname}{Acknowledgements}
\begin{abstract}
The author would like to thank Prof. Peter Pfaffelhuber from the department of statistics, University of Freiburg, for making helpful comments on the manuscript.
\end{abstract}

\bibliographystyle{chicago}
\bibliography{lit_self-regulation-gene}

\end{document}